\documentclass[aoas,preprint]{imsart}

\RequirePackage{amsthm,amsmath,amsfonts,amssymb}
\RequirePackage[authoryear]{natbib}

\startlocaldefs

\usepackage[utf8]{inputenc} 
\usepackage[T1]{fontenc}    %
\usepackage[english]{babel} 
\usepackage[final]{microtype} 
\usepackage{mathtools,bm}

\usepackage{graphicx} 
\usepackage[dvipsnames]{xcolor}

\usepackage{booktabs,multirow}

\usepackage[normalem]{ulem}

\usepackage{tabularx}
\newcolumntype{b}{>{\hsize=1.2\hsize}X}
\newcolumntype{s}{>{\hsize=.45\hsize}X}
\newcolumntype{m}{>{\hsize=.8\hsize}X}

\usepackage{pgf,tikz}
\usetikzlibrary{arrows,shapes.arrows,shapes.geometric,shapes.multipart,
  decorations.pathmorphing,positioning,
  arrows.meta}

\tikzset{node distance = 1cm and 1.5cm}
\tikzset{>={Latex[width=2mm,length=2mm]}} 

\usepackage{enumitem}
\setenumerate{label=(\roman*)}

\usepackage{afterpage}
\usepackage{pdflscape}

\usepackage{subcaption}

\allowdisplaybreaks

\usepackage{hyperref}   

\usepackage{memhfixc}   %
\usepackage{cleveref}

\newcommand{\todo}[1]{}
\renewcommand{\todo}[1]{{\color{Red} TODO: {#1}}}

\newcommand\independent{\protect\mathpalette{\protect\independenT}{\perp}}
\def\independenT#1#2{\mathrel{\rlap{$#1#2$}\mkern2mu{#1#2}}}

\let\E\relax
\DeclareMathOperator{\E}{\mathbb{E}}
\let\P\relax
\DeclareMathOperator{\P}{\mathbb{P}}

\newcommand{\approptoinn}[2]{\mathrel{\vcenter{
  \offinterlineskip\halign{\hfil$##$\cr
    #1\propto\cr\noalign{\kern2pt}#1\sim\cr\noalign{\kern-2pt}}}}}

\newcommand{\appropto}{\mathpalette\approptoinn\relax}

\usepackage{mdframed}

\theoremstyle{theorem}

\crefname{theorem}{Theorem}{Theorems}

\newtheorem{proposition}{Proposition}
\crefname{proposition}{Proposition}{Propositions}

\newtheorem{lemma}{Lemma}
\crefname{lemma}{Lemma}{Lemmas}

\crefname{corollary}{Corollary}{Corollarys}

\theoremstyle{definition}

\crefname{example}{Example}{Examples}

\crefname{exercise}{Exercise}{Exercises}

\crefname{definition}{Definition}{Definitions}

\newtheorem{assumption}{Assumption}
\crefname{assumption}{Assumption}{Assumptions}

\theoremstyle{remark}

\crefname{remark}{Remark}{Remarks}

\newcommand{\rev}[1]{{\color{Black} #1}}
\newcommand{\revv}[1]{{\color{Black} #1}}
\newcommand{\aoasins}[1]{#1}
\newcommand{\aoasdel}[1]{}

\endlocaldefs

\begin{document}

\begin{frontmatter}
\title{BETS: The dangers of selection bias in early analyses of the
  coronavirus disease (COVID-19) pandemic\thanksref{T1}}
\runtitle{BETS in COVID-19}
\thankstext{T1}{Accepted manuscript.}

\begin{aug}

\author[A]{\fnms{Qingyuan} \snm{Zhao}\ead[label=e1]{qyzhao@statslab.cam.ac.uk}},
\author[B]{\fnms{Nianqiao} \snm{Ju}\ead[label=e2,mark]{nju@g.harvard.edu}},
\author[A]{\fnms{Sergio}
  \snm{Bacallado}\ead[label=e3,mark]{sb2116@cam.ac.uk}}
\and
\author[A]{\fnms{Rajen D.}
  \snm{Shah}\ead[label=e4,mark]{R.Shah@statslab.cam.ac.uk}}

\address[A]{Statistical Laboratory,
  Department of Pure Mathematics and
  Mathematical Statistics,
  University of Cambridge,
  \printead{e1,e3,e4}}
\address[B]{Department of Statistics,
  Harvard University,
  \printead{e2}}

\end{aug}

\begin{abstract}
  ~The coronavirus disease 2019 (COVID-19) has quickly grown from a
  regional outbreak in Wuhan, China to a global pandemic. Early
  estimates of the epidemic growth and incubation period of COVID-19
  may have been biased due to sample selection. Using
  detailed case reports from 14 locations in and outside mainland China,
  we obtained 378 Wuhan-exported cases who left Wuhan before an abrupt
  travel quarantine. We developed a generative model we call BETS for four key
  epidemiological events---Beginning of exposure, End of exposure,
  time of Transmission, and time of Symptom onset (BETS)---and derived
  explicit formulas to correct for the sample selection. We gave a
  detailed illustration of why some early and highly influential
  analyses of the COVID-19 pandemic were severely biased. All our
  analyses, regardless of which subsample and model were being used,
  point to an epidemic doubling time of $2$ to $2.5$ days during the
  early outbreak in Wuhan. A Bayesian nonparametric analysis further
  suggests that about 5\% of the symptomatic cases may not develop symptoms
  within 14 days of infection and that men may be much more likely
  than women to develop symptoms within 2 days of infection.
\end{abstract}


\end{frontmatter}





\section{Introduction}
\label{sec:introduction}

On December 31, 2019, the Health Commission in Wuhan, China, announced 27 cases of
unknown viral pneumonia and alerted the World Health Organization. The causative pathogen
was quickly identified as a novel coronavirus and the disease was later designated as the
coronavirus disease 2019 (COVID-19) \citep{who2020statement}. The
regional outbreak in Wuhan quickly turned into a global pandemic. As
of June 1, 2020, COVID-19 has reached almost every country in the
world, infected more than 6 million people, and killed at least 370,000
\citep{jhucovid19}.

Researchers around the world quickly responded to the COVID-19
outbreak. In particular, many have examined early outbreak data to
estimate the initial epidemic growth, using COVID-19 cases confirmed
in Wuhan or elsewhere. Two early studies published in premier medical
journals by the end of January estimated that the epidemic doubling
time in Wuhan was about $6$ to $7$ days
\citep{li2020early,wu2020nowcasting}, but other studies appearing around the
same time found that the doubling time was drastically shorter, about $2$ to
$3$ days \citep{read2020novel,sanche2020high,zhao2020analysis}. How the
pandemic subsequently developed around the world seems to suggest that
the latter estimates were much closer to truth. By simply plotting the
cumulative cases and deaths over time, it is evident now that the
number of cases (and deaths) grew more than
100 times 20 days after the first 100 cases (and 10 deaths) in
countries most heavily hit by the pandemic such as Italy, Spain,
and the United States (\Cref{fig:growth}). That growth rate almost exactly
corresponds to a doubling time of 3 days. Nevertheless, to our
knowledge there is no formal explanation for this drastic difference, and
it might have caused confusion during the early phase of
containment of COVID-19. For example, during the UK government's daily
briefing on March 16, it was acknowledged that ``without drastic
action, cases could double every five or six days''
\citep{ukmarch16}. Less than two weeks later, that number was revised
to ``three to four days'' \citep{govemarch27}.

For infectious diseases, another key epidemiological parameter is the
incubation period. Several studies have attempted to estimate the
incubation period distribution of COVID-19 using cases exported from
Wuhan
\citep{backer2020incubation,lauer2020incubation,linton2020incubation}
and the results have been influential in shaping guidelines to manage
confirmed COVID-19 patients. For example, the interim clinical
guidance for managing COVID-19 patients published by the Centers for
Disease Control and Prevention (CDC) \citep{cdc2020guidance} quoted
the results of \citet{lauer2020incubation} that ``97.5\%
of persons with COVID-19 who develop symptoms will do so within 11.5
days of SARS-CoV-2 infection.''  However, as we will demonstrate below in
\Cref{sec:why-prev-analys}, the design and statistical inference of
these studies are highly susceptible to selection bias.

\begin{figure}[t]
  \centering
  \begin{subfigure}[b]{0.49\textwidth}
    \includegraphics[width=\textwidth]{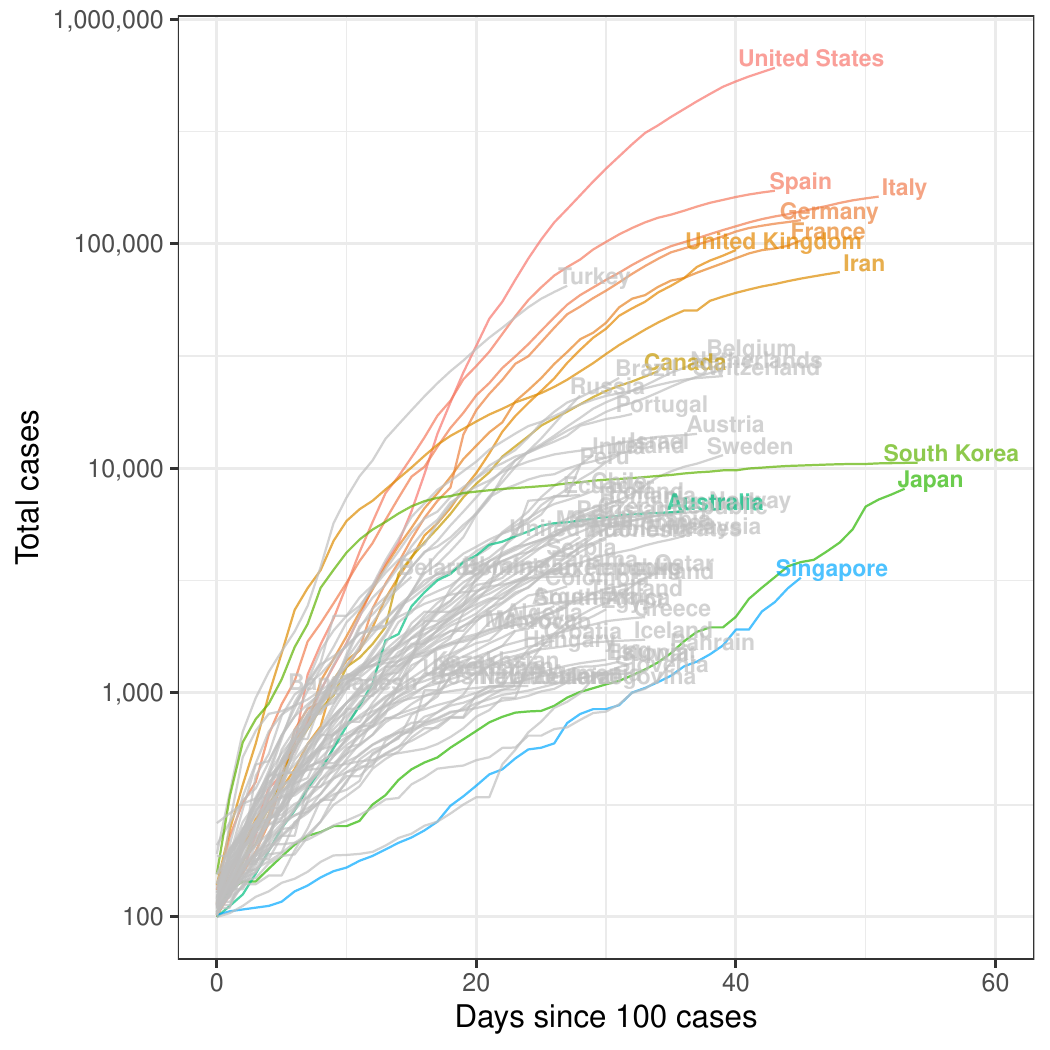}
    \caption{Cumulative cases.}
  \end{subfigure}
  \begin{subfigure}[b]{0.49\textwidth}
    \includegraphics[width=\textwidth]{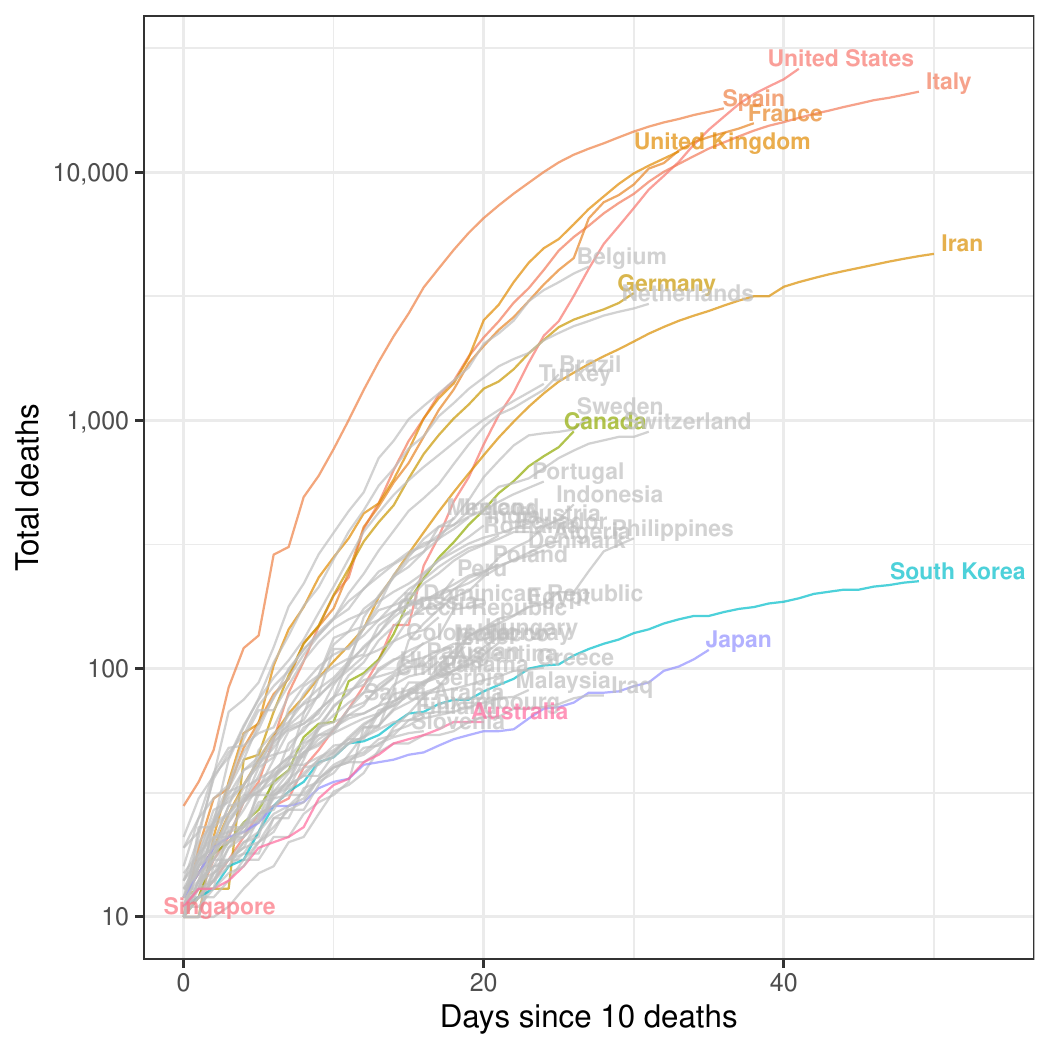}
    \caption{Cumulative deaths.}
  \end{subfigure}
  \caption{Growth of the COVID-19 pandemic around the world (data
    retrieved from \url{https://www.worldometers.info/} on April 15, 2020).}
  \label{fig:growth}
\end{figure}

Accurate estimation of these epidemiological parameters are crucial to the
global efforts to contain and mitigate the COVID-19 pandemic. Besides
direct implications in patient management and quarantine policies,
they are also inputs to the compartmental models
\citep{kermack1927contribution,hethcote2000mathematics} and
individual-based models \citep{willem2017lessons} to forecast the
epidemic and inform public health policies. In general, there are
several potential sources of bias in early analyses of the COVID-19
pandemic \revv{(see also \Cref{tab:summary-bias})}:
\begin{enumerate}
\item {\bfseries Under-ascertainment:} Because COVID-19 is a new
  disease, the testing capacity was very limited during the early
  stage of the outbreak. The eligibility criterion for testing was
  initially very strict. This may explain why \citet{li2020early}
  under-estimated the epidemic growth as they only used cases in Wuhan who
  showed symptoms before January 5, 2020.
\item {\bfseries Non-random sample selection:} \revv{Not all public health
    agencies reported detailed information of COVID-19 cases. Many
    stopped doing so after the first few cases. Studies which only
    collect complete or conveniently available data may be biased by
    non-random sample selection.} For example, it is
  often impossible to know the exact time when the cases were
  infected. If one simply uses cases with known infection time, the
  incubation period may be under-estimated because it is more
  difficult to discern the infection time for cases with longer
  incubation period.
\item {\bfseries Travel quarantine:} Wuhan is a major transportation
  hub in central China. To control the spread of the virus, all
  outbound travels from Wuhan were abruptly halted on January 23,
  2020. For studies using cases exported from Wuhan \revv{(COVID-19
    cases who were infected in Wuhan and confirmed elsewhere)},
  ignoring the sample selection due to the travel quarantine leads to
  biased estimates of the epidemic growth.
  \item \revv{ {\bfseries Ignoring epidemic growth:} Because the epidemic was
    rapidly growing, patients were more likely to be infected towards
    the end of their exposure period. Ignoring the growth and using a simple
    uniform distribution for the infection time over a prolonged exposure
    period may lead to over-estimation of the incubation period.}
\item \revv{ {\bfseries Right-truncation:} Early analyses of the
    epidemic were limited to using cases confirmed before a certain
    date, when the number of infections was still growing rapidly. This
    may lead to under-estimation of
    the incubation period, as people with milder symptoms or longer
    incubation period are less likely to be included in the study.}
\end{enumerate}

In this article, we address these challenges by carefully constructing a
study sample and a statistical model. We collected key epidemiological
information for 1,460 confirmed COVID-19 cases across 14 locations in
and outside mainland China. By focusing on locations where the
local health agencies made great efforts to contain the initial
outbreaks and published detailed case reports, the biases due to (i)
under-ascertainment and (ii) non-random selection are
minimized. \Cref{sec:data} describes how our data were collected and
the Wuhan-exported cases were discerned.

\revv{
We addressed potential biases due to (iii) the travel quarantine, (iv)
ignoring epidemic growth, and (v) right-truncation by constructing a
generative statistical model. We call it
the BETS model, as it models four key epidemiological events:
Beginning of exposure, End of exposure, time of Transmission, and time of
Symptom onset. The travel quarantine puts a constraint on the support
of the observed data for Wuhan-exported cases, for which we carefully
worked out the selection
probability and used it to adjust the likelihood function. Epidemic
growth is naturally considered in the estimation of the incubation
period because they are estimated jointly using the likelihood functions
we derived. Sample selection due to right-truncation can also be
characterized and adjusted for. Details of the generative model and
likelihood inference can be found in \Cref{sec:model}.
}

We then give a detailed explanation in \Cref{sec:why-prev-analys} of why
some early analyses of the
COVID-19 outbreak were severely biased, including the estimation of epidemic growth
by \citet{wu2020nowcasting} and the estimation of incubation period by
\citet{backer2020incubation,lauer2020incubation,linton2020incubation}. Because
these analyses did not start from a generative model, they could not
correctly adjust for sample selection in their statistical
inference.

In order to obtain closed-form likelihood functions in
\Cref{sec:model}, we introduced some parametric assumptions
which necessarily restrict the shape of the tail of the incubation
period distribution. To avoid
biased tail estimates, we model the distribution nonparametrically and
also relax the other assumptions in
\Cref{sec:bayes}. Because the likelihood function is no longer
available in closed form, a Markov Chain Monte Carlo (MCMC) sampler is
needed for Bayesian nonparametric inference. Finally, we summarize our
findings and discuss potential limitations of our study in
\Cref{sec:discussion}. All technical derivations can be found in the
Online Supplement; our dataset and statistical programs are publicly available
as an \texttt{R} package from
\url{https://github.com/qingyuanzhao/bets.covid19}.

  \begin{landscape} \centering

    \begin{table}[t] 
      \centering
      \caption{Summary of potential biases in analyses of the
        COVID-19 pandemic.}
      \label{tab:summary-bias}
      \begin{tabularx}{\linewidth}{b|mmb}
        \toprule
        \textbf{Bias} & \textbf{Susceptible studies} & \textbf{Direction} & \textbf{Solutions} \\
        \midrule
        \textbf{(i) Under-ascertainment:} Symptomatic patients
        did not seek healthcare or could not be diagnosed. & All studies
        using cases confirmed when testing is insufficient. &
        \textbf{Varied,} depending on the pattern of
        under-ascertainment and parameter of interest. & Use
        carefully considered and planned study designs. \\
        \midrule
        \textbf{(ii) Non-random sample selection}: Cases included in
        the study are not representative of the population. &
        All studies, as detailed information of COVID-19 cases is
        sparse, but especially those without clear inclusion
        criteria. &
        \textbf{Varied.} & Follow a protocol for data collection
        \aoasdel{exclude data that do not meet the sample
          inclusion criterion.}\aoasins{with
          a clearly defined sample inclusion criterion.}
        \\
        \midrule
        \textbf{(iii) Travel quarantine:} Outbound
        travel from Wuhan was banned from January 23, 2020 to April
        8, 2020. & Studies that analyze cases exported from Wuhan. &
        \textbf{Under-estimation} of epidemic growth \citep{wu2020nowcasting} and
        infection-to-recovery time \citep{dorigatti2020report}. & Derive
        tailored likelihood functions to account for travel
        restrictions. (See \Cref{sec:estim-epid-growth}.) \\
        \midrule
        \textbf{(iv) Epidemic growth:} Patients were more
        likely to be infected towards the end of their exposure
        period. & Studies that treat
        infections as uniformly distributed over the exposure period.
        & \textbf{Over-estimation} of incubation period
        \citep{backer2020incubation,lauer2020incubation,linton2020incubation}
        and serial interval
        \citep{du2020serial,nishiura2020serial}. & Derive
        tailored likelihood functions to account for epidemic
        growth. (See \Cref{sec:estim-incub-peri}.) \\
        \midrule
        \textbf{(v) Right-truncation:} Cases confirmed after a certain time are
        excluded from the dataset. & Studies that only use cases detected
        early in an epidemic. & \textbf{Under-estimation} of incubation
        period
        \citep{backer2020incubation,lauer2020incubation,linton2020incubation},
        serial interval \citep{du2020serial,nishiura2020serial}, and disease severity. & \textbf{1.\ }Collect all
        cases that meet a selection criterion, do not end data collection
        prematurely; \textbf{2.\ }Derive tailored likelihood functions to
        correct for right-truncation. (See
        \Cref{sec:estim-incub-peri}.) \\
        \bottomrule
      \end{tabularx}
    \end{table}
  \end{landscape}



\section{Data}
\label{sec:data}

\subsection{Data Collection}
\label{sec:data-collection}

\begin{figure}[t]
  \centering
  \includegraphics[width = 0.7\textwidth]{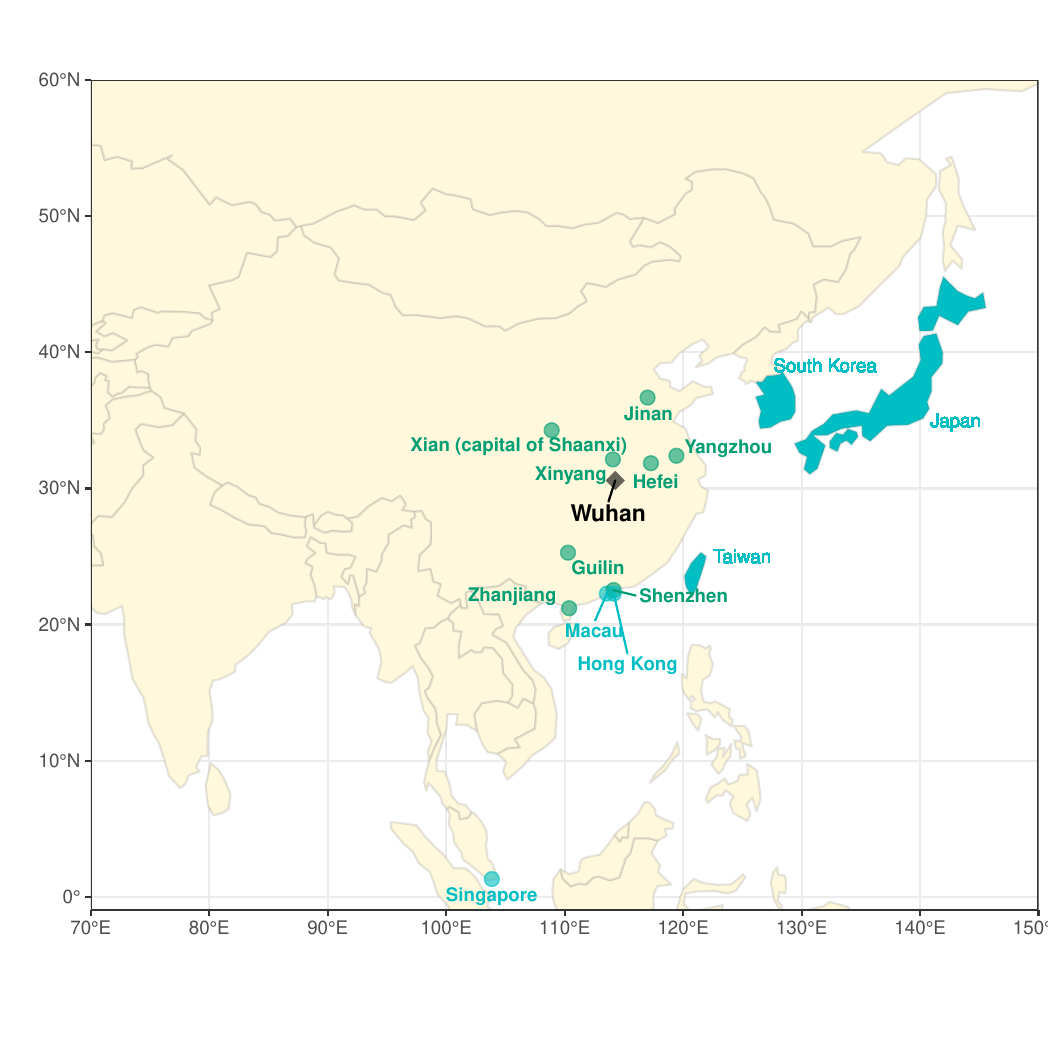}
  \caption{Geographical locations of the confirmed cases in our dataset.}
  \label{fig:map}
\end{figure}

We identified 14 locations where the local health agencies have published
continuous reports for every confirmed COVID-19 case since the first
local case. Out of the 14 locations, 8 are cities/provinces in
mainland China: Hefei, Guilin, Jinan, Shaanxi, Shenzhen, Yangzhou,
Xinyang, Zhanjiang and 6 are countries/regions in East Asia: Hong Kong, Japan,
South Korea, Macau, Singapore, and Taiwan (\Cref{fig:map}). These
locations have varied levels of economic development and patterns of
traveling to/from Wuhan. Key information (close contact, travel history,
symptom onset) of the confirmed COVID-19 cases was collected based
on press releases of the official health agencies (\Cref{tab:data}). In total,
there are 1,460 COVID-19 cases in the collected dataset.

\begin{table}[t] \scriptsize
  \centering
    \caption{A summary of the key columns in the collected
    dataset. \revv{Boxed entries correspond to the recorded values of the example
      (HongKong-05).} 
    {
    $^1$Description of this case in Hong Kong government's press
    release on January 24, 2020: ``The other two cases are a married couple of
    residents of in Wuhan, a 62-year-old female [HongKong-04] and a
    63-year-old male [HongKong-05], with good prior health conditions. Based
    on information provided by the patients, They took a high-speed train
    departing from Wuhan at 2:20pm, January 22, and arrived at the
    West Kowloon station around 8pm. The female patient had a fever
    since yesterday with no respiratory symptoms. The male patient
    started to cough yesterday and had a fever today. They went to the
    emergency department at the Prince of Wales Hospital yesterday and
    were admitted to the hospital for treatment in isolation. Currently
    their health conditions are stable. Respiratory samples of the two
    patients were tested positive for the novel coronavirus.''
    (translated from Chinese). 
    $^2$A case is considered to have known epidemiological contact if he/she
    had contact with people from the Hubei province or had contact
    with another case confirmed earlier.
    $^3$See the main text for
    the criterion we used to classify the cases. 
    $^4$The beginning
    of stay is treated as November 30 if the case resides in Wuhan
    and has no known beginning of stay.}}
  \label{tab:data}

  \begin{tabular}{l|l|c|c}
    \toprule
    \textbf{Column name} & \textbf{Description} & \textbf{Example$^1$} & \textbf{Summary statistics} \\
    \midrule
    \texttt{Case} & Unique identifier for each case &
                                                            HongKong-05
                               & 1460 in total \\
    \texttt{Residence} & Nationality or residence of the case & Wuhan
                                                 & 21.5\% reside in Wuhan  \\
    \texttt{Gender} & Gender & \fbox{Male}/Female & 52.1\%/47.7\%
                                                    (0.2\% unknown) \\
    \texttt{Age} & Age & 63 & Mean=45.6, IQR=[34, 57] \\
    \midrule
    \texttt{Known Contact} & Have known epidemiological contact$^2$? &
                                                                       \fbox{Yes}/No & 84.7\%/15.3\% \\
    \texttt{Cluster} & Relationship with other cases & Husband of &
                                                                    32.1\%
                                                                    known \\
    & & HongKong-04 & \\
    \texttt{Outside} & Transmitted outside Wuhan?$^3$ &
                                                        Yes/\fbox{Likely}/No
                                                 & 58.5\%/7.7\%/33.8\%\\
    \midrule
    \texttt{Begin Wuhan} & Begin of stay in Wuhan ($B$) & 30-Nov$^4$ &  \\
    \texttt{End Wuhan} & End of stay in Wuhan ($E$) & 22-Jan &  \\
    \texttt{Exposure} & Period of exposure & 1-Dec to 22-Jan & 58.9\%
                                                               known
                                                               period/date
    \\
                         & & & 8.2\% known date \\
    \midrule
    \texttt{Arrived} & Final arrival date at the location
                               & 22-Jan & 40.6\% did not travel outside \\
    & where confirmed a COVID-19 case & & \\
    \texttt{Symptom} & Date of symptom onset ($S$) & 23-Jan & 9.0\% unknown \\
    \texttt{Initial} & Date of first medical visit/quarantine & 23-Jan
                                                 & 6.5\% unknown  \\
    \texttt{Confirmed} & Date confirmed as a COVID-19 case & 24-Jan &  \\
    \bottomrule
  \end{tabular}

\end{table}

\begin{figure}[t]
  \centering
  \includegraphics[width = \textwidth]{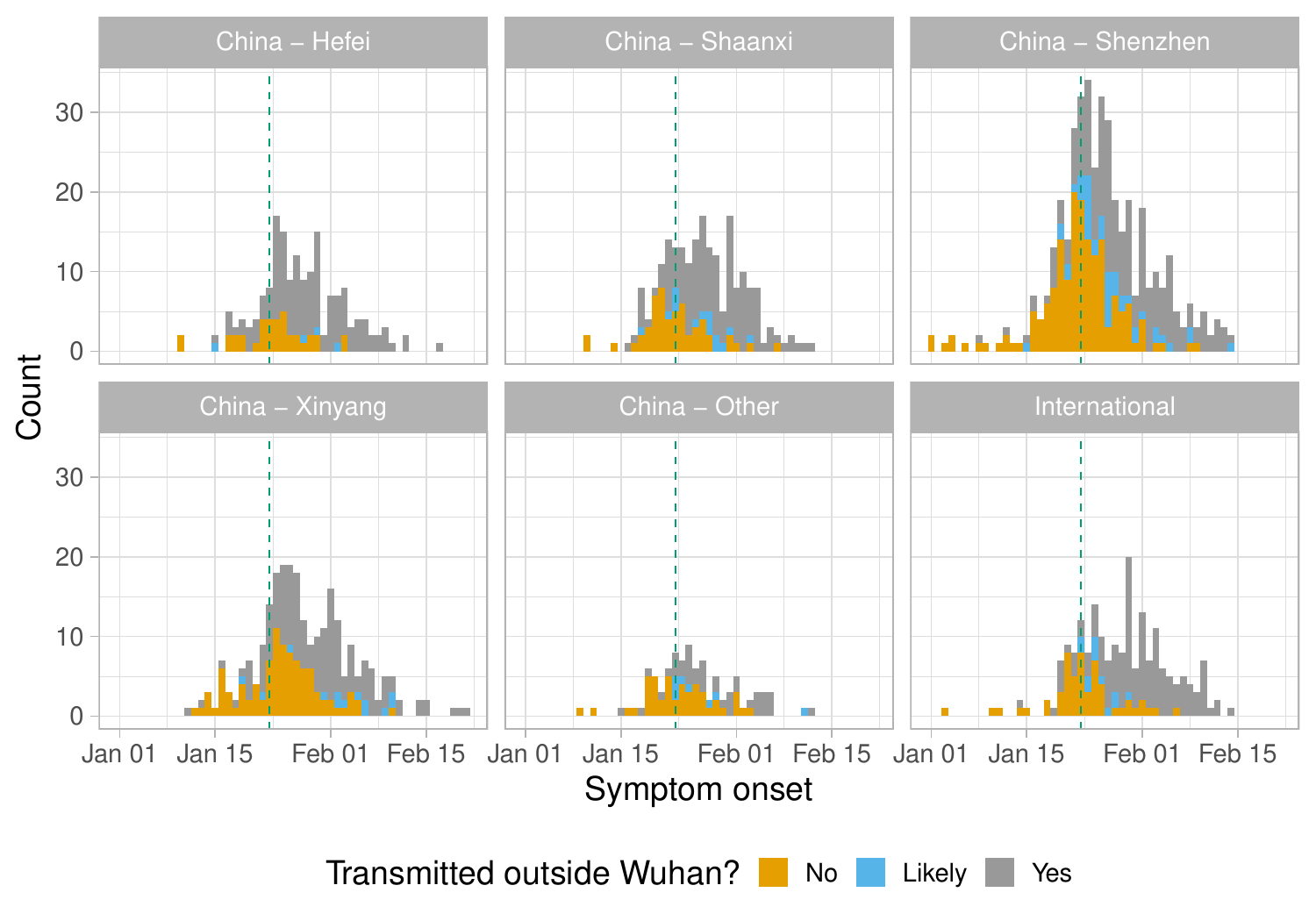}
  \caption{Epidemic curves in different locations stratified by
    whether the cases were transmitted outside Wuhan. ``China - Other''
    includes four Chinese cities: Guilin, Jinan, Yangzhou, Zhanjiang;
  ``International'' includes six Asian countries/regions: Hong Kong,
  Japan, Korea, Macau, Singapore, Taiwan. The dashed vertical lines
  correspond to the abrupt travel quarantine of Wuhan from January 23,
  2020.}
  \label{fig:epidemic-curve}
\end{figure}

For the mainland Chinese locations, the dataset included all the cases
confirmed as of February 29, 2020. In Chinese cities outside the Hubei
province, local epidemics were considered to be successfully contained
by the end of February. For the international locations, the dataset
included all the cases confirmed before February 15, more than three
weeks after the outbound travel quarantine of Wuhan on January
23. It is thus safe to say that our dataset contains
almost all Wuhan-exported cases confirmed in these locations.

\subsection{Discerning Wuhan-exported cases}
\label{sec:wuhan-exported-cases}

\revv{We define Wuhan-exported cases as those who were infected in
  Wuhan and confirmed elsewhere.}
In total, $614$ cases in our dataset are potentially exported from
Wuhan because they had stayed in Wuhan before got diagnosed
elsewhere. Because Wuhan was the first center of epidemic
outbreak and traveling from/to Wuhan was not restricted before
January 23, it is reasonable to assume that most of these
$614$ cases were infected there. However, some uncertainty arises
if a case had contact with other confirmed cases outside their stay in
Wuhan, in one of the following scenarios:
\begin{itemize}
\item The case already had contact with other confirmed cases before
  their stay in Wuhan ($4$ cases);
\item The case had contact with other confirmed cases only after they
  left Wuhan but before they arrived at their destination, for example
  in trains or flights ($4$ cases);
\item The case had close contact with other confirmed cases (usually
  family members \revv{who traveled together} from Wuhan) after they reached
  their travel destination ($131$ cases).
\end{itemize}

\revv{To discern a dataset of Wuhan-exported cases, the principle we
followed is to exclude cases if there is a ``reasonable doubt'' that they
could be infected outside Wuhan.} We assumed in the first two scenarios
above, the cases were transmitted outside Wuhan. For the third
scenario, it is likely that the cases were
transmitted outside Wuhan, but at least one of the cases in each cluster
were transmitted in Wuhan. (Two cases are considered to belong to the
same cluster if they are in the family or had other recorded contact.)
We used a column called \texttt{Outside} in our dataset to record our
best judgment on whether the cases were transmitted outside Wuhan
using the following rules:
\begin{enumerate}
\item \texttt{Outside} = ``Yes'': Cases with no recorded stay in Wuhan between
  December 1, 2019 and January 23, 2020, and the $8$ cases in the
  first two scenarios above ($854$ cases).
\item \texttt{Outside} = ``Likely'': Wuhan-exposed cases who did not
  show symptoms during the recorded stay in Wuhan and had recorded
  contact with another confirmed COVID-19 case with an earlier symptom
  onset ($112$ cases).
\item \texttt{Outside} = ``No'': Wuhan-exposed cases who had no recorded
  contact with other confirmed cases, or had the earliest symptom
  onset in their cluster or showed symptoms during
  their stay in Wuhan ($494$ cases).
\end{enumerate}
\Cref{fig:epidemic-curve} shows the local epidemic curves stratified
by the \texttt{Outside} column in different locations.

The dataset we collected has relatively few missing values in the key
entries needed for epidemic modeling. Among the \texttt{Outside} =
``No'' cases, only 6.5\% do not have the exact date they left Wuhan
and only 8.1\% have missing symptom onset date (including
those showing no symptoms at the time of confirmation). \revv{We
  imputed the missing end of exposure to Wuhan by the day of the
  travel quarantine (January 23) and excluded the cases with missing
  symptom onset. This left us with 458 cases. We further excluded
  cases who arrived at the location where they are diagnosed with
  COVID-19 after January 23 as they have a different traveling pattern
than the other cases. In the end we obtained 378 cases who were
exported from Wuhan.}


\section{Statistical model and parametric inference}
\label{sec:model}


\subsection{BETS: A generative model}
\label{sec:generative-model}

\revv{On a high level, our goal is to make inference about the epidemic in
Wuhan using its ``shadows'' observed in other locations. To properly
consider the consequences of sample selection,} we will first outline
a generative model for (which is also named after) four key epidemiological events: the
beginning of stay in Wuhan $B$, the end of stay in Wuhan $E$, the
usually unobserved time of transmission $T$, and the time of
symptom onset $S$ (BETS). These four variables are well defined
regardless of whether the person has been to Wuhan, contracted the
pathogenic coronavirus, or showed symptoms of COVID-19. \aoasins{Even
  though the observed times are discrete (in days), we
  will model $B,E,T$, and $S$ as continuous random variables
  (with $\infty$ reserved for special events). This allows us to
  simplify the interpretation of the modeling assumptions and obtain
  closed-form solutions. When applying the likelihood functions below
  on the discretized observations, the results may be interepreted as an
  approximation to the model assuming $B,E$, and $S$ are uniformly
  distributed around the discretized measurements. See
  \Cref{sec:impl-likel-infer} for details of the implementation and
  \Cref{sec:bayes} for how one can directly model the
  discretized times.}

\subsubsection*{Study population: Exposed to Wuhan}
\label{sec:support-variables}

Consider the population of all people who stayed in
Wuhan any time between 12AM December 1, 2019 (time $0$) and 12AM
January 24, 2020 (time $L$ when outbound travel from Wuhan was banned,
$L = 54$) in local time. The choice time $0$, as long as it is
reasonably early, has essentially no effect on our analysis below,
because almost all of the transmissions happened in January. We
introduce the following conventions to define the population with
exposure to Wuhan:

\begin{itemize}
\item $B = 0$: The person started their
  stay in Wuhan before December 2019.
\item $E = \infty$: The person did not arrive in the 14 locations we
  are considering before the travel quarantine (time $L$). \revv{For
    the purpose of this study, we need not differentiate between
    people who stayed in Wuhan or went to a location different from
    the ones we are considering}.
\item $T = \infty$: The person did not contract the pathogenic virus
  during their stay in Wuhan. For the purpose of this study, we
  need not differentiate between people who contracted the virus
  outside their Wuhan stay and people who never contracted the virus.
\item $S = \infty$: The person did not show symptoms of COVID-19,
  either because they never contracted the virus or were
  asymptomatic.
\end{itemize}

Because we are only considering people exposed to Wuhan, we
have $B \le L$. Two other natural constraints are $B \le E$ and $T \le
S$ (where we allow $\infty \le \infty$). Therefore, the support of
$(B,E,T,S)$ for the Wuhan-exposed population is
\begin{equation}
  \label{eq:population}
  \mathcal{P} = \Big\{(b,e,t,s) \mid b \in [0,L], e \in
  [b,L] \cup \{\infty\}, t \in [b,e] \cup \{\infty\}, s \in [t,\infty]\Big\}.
\end{equation}

Notice that although $B$ \aoasdel{is supported on}\aoasins{has support} $[0,L]$, $B=0$ is a point
mass representing Wuhan residents and is categorically different from
$B = \epsilon$ for some small positive $\epsilon$. All density
functions of $B$ below are defined with respect to the sum of Lebesgue
measure on the real line and degenerate counting measure for
$\{0\}$. Similarly, for $E$, $T$, and $S$ the dominating measure is
the sum of Lebesgue measure and the counting measure for
$\{\infty\}$. Joint densities of $(B,E,T,S)$ below are defined with
respect to their product measure.

\subsubsection*{Full data BETS model\revv{: Independence of
    traveling and disease transmission/progression}}

\revv{In \Cref{sec:sample-selection} below we will define the constraints
  corresponding to our sample selection. But first, we will introduce
  a generative statistical model for $(B,E,T,S)$ in the whole study
  population $\mathcal{P}$.} The joint density of $(B,E,T,S)$ can
always be factorized as:

\begin{equation}
  \label{eq:joint-density}
  f(b,e,t,s) = f_B(b) \cdot f_E(e \mid b) \cdot f_T(t \mid
  b, e) \cdot f_S(s \mid b,e,t).
\end{equation}

Throughout this article we will maintain two general assumptions about
two conditional densities in this factorization:

\begin{assumption} \label{assump:model-t}
  The conditional density $f_T(t \mid b,e)$ does not depend on $b$ and
  $e$ in the range $b \le t \le e$, so it can be written as
  \begin{equation}
    \label{eq:model-t}
    f_T(t \mid b,e) =
    \begin{cases}
      g(t), & \text{if}~b < t < e, \\
      1 - \int_b^e g(x)\, dx, & \text{if}~t = \infty.
    \end{cases}
  \end{equation}
  Here $g(t) \ge 0$ models the epidemic growth in Wuhan before the
  citywide quarantine on January 23; it can be interpreted as the
  instantaneous probability of being infected in Wuhan at time $t$ and
  satisfies the constraint $\int_{0}^L g(x) \, dx \le 1$.
\end{assumption}

\begin{assumption} \label{assump:model-s}
  The conditional density $f_S(s \mid b,e,t)$ does not depend on $b$
  and $e$, so it can be written as
  \begin{equation}
    \label{eq:model-s}
    f_S(s \mid b,e,t) = \begin{cases}
      \nu \cdot h(s - t),& \text{if}~s < \infty, \\
      1-\nu, & \text{if}~s = \infty.
    \end{cases}
  \end{equation}
  Here \aoasins{$1-\nu$ is the probability of asymptomatic infection,}
  $h(s-t)$ is the conditional density of the incubation
  period $S - T$ given that $S - T < \infty$ (the case is not
  asymptomatic), so $h(\cdot)$ satisfies $\int_0^{\infty} h(x) \, dx
  = 1$.
\end{assumption}

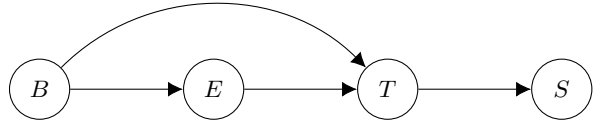
\begin{figure}[t]
  \centering
  \begin{tikzpicture}[minimum size = 0.8cm]
    \node[circle,draw,name=b]{$B$};
    \node[circle,draw,name=e,right= of b]{$E$};
    \node[circle,draw,name=t,right= of e]{$T$};
    \node[circle,draw,name=s,right= of t] {$S$};
    \draw[->] (b) to (e);
    \draw[->,out=45,in=135] (b) to (t);
    \draw[->] (e) to (t);
    \draw[->] (t) to (s);
  \end{tikzpicture}
  \caption{Directed acyclic graph (DAG) for the BETS model. $B$ is the
  beginning of exposure, $E$ is the end of exposure, $T$ is the time
  of transmission, and $S$ is the time of symptom onset.}
  \label{fig:bets-dag}
\end{figure}

\Cref{assump:model-t,assump:model-s} essentially mean that the
disease transmission and progression are independent of
traveling, which allows us to extend conclusions learned from the
Wuhan-exported sample to the whole population. \Cref{assump:model-s}
is equivalent to the conditional independence $S \independent (B,E)
\mid T$, which can be represented as a directed acyclic graphical
(DAG) model (\Cref{fig:bets-dag}) on
the distribution of $(B,E,T,S)$
\citep{lauritzen1996graphical}. \Cref{assump:model-t} further
restricts the dependence of $T$ on $(B,E)$. Under these two
assumptions, the BETS model is then
parameterized by two kinds of parameters: \revv{the nuisance} parameters for the traveling
pattern $f_B(\cdot)$ and $f_E(\cdot \mid \cdot)$, and \revv{the}
parameters \revv{of interest} for disease transmission $g(\cdot)$ and
progression $h(\cdot)$.

Like any other assumptions in epidemic models,
\Cref{assump:model-t,assump:model-s} represent approximations to the
underlying
dynamics. \Cref{assump:model-t,assump:model-s} can be violated if, for
example, short-term visitors were exposed to more
infectious cases or if people were less likely to travel if they felt
sick. Nevertheless, we think they are reasonable approximations to
the reality during the initial outbreak, when little was known about
the new infectious disease.

\subsubsection*{Parametric assumptions \revv{for closed-form
    likelihood functions}}

\Cref{assump:model-t,assump:model-s} are general assumptions
on the dependence of $T$ and $S$ on $B$ and $E$. We consider two
parametric assumptions that simplify the interpretation of our
results:
\begin{assumption} \label{assump:model-t-parametric}
  The probability of contracting the virus in Wuhan was increasing
  exponentially before the quarantine:
  \begin{equation}
    \label{eq:model-t-parametric}
    g(t) = g_{\kappa,r}(t) \overset{\Delta}{=} \kappa \cdot \exp(rt),~t \le L,
  \end{equation}
  where \aoasdel{$(\kappa,r)$}\aoasins{the baseline attack rate $\kappa$
    and the growth exponent $r$} satisfies $\int_0^L g_{\kappa,r}(t)
  \, dt \le 1$.
\end{assumption}

\begin{assumption}  \label{assump:model-s-parametric}
  The incubation period $T-S$, given that it is finite (the case is
  not asymptomatic), follows a Gamma distribution with shape
  $\alpha > 0$ and rate $\beta > 0$:
  \begin{equation}
    \label{eq:model-s-parametric}
    h(s - t) = h_{\alpha,\beta}(s-t) \overset{\Delta}{=} \frac{\beta^{\alpha}}{\Gamma(\alpha)} (s-t)^{\alpha-1}
    \exp\{-\beta (s-t)\}.
  \end{equation}
\end{assumption}

\Cref{assump:model-t-parametric} says that the epidemic size in
Wuhan was growing exponentially before the quarantine, which is a
common assumption for early epidemic outbreaks. We think it is quite
reasonable given that little was known about the novel coronavirus
before January 23. \aoasins{One of the main objectives in early
  epidemic analysis is to estimate the growth exponent $r$, which can be
  translated to the doubling time by finding the $t$ such that
  $\exp(rt) = 2$ (i.e.\ $t = \log(2)/r$). The
  parameter $\kappa$ is important for measuring the absolute
  magnitude of the epidemic. It is possible to estimate $\kappa$ using
  Wuhan-exported cases if we have external data on the number of
  travellers from Wuhan and the total population of Wuhan
  \citep{imai2020report1,wu2020nowcasting}. This paper
  will focus on the estimation of $r$; in fact, we will see below that
  the selection-adjusted likelihood functions for Wuhan-exported cases
  do not depend on $\kappa$.}

\Cref{assump:model-s-parametric} restricts the
density function $h(\cdot)$ to the Gamma family, which is commonly
used to model the distribution of the incubation
period. \Cref{assump:model-t-parametric,assump:model-s-parametric}
will be used later in this and the next sections to
calculate closed-form likelihood functions. Later in \Cref{sec:bayes}, we will
relax the parametric assumptions to allow more flexible patterns for
the epidemic growth and more general distributions of the incubation
period.

\subsection{Accounting for sample selection in the likelihood}
\label{sec:sample-selection}

\subsubsection*{Study sample: Wuhan-exported cases}
\label{sec:study-sample}

To use Wuhan-exported cases to study the epidemic growth and
incubation period, it is crucial to consider the effect of sample
selection on Wuhan-exported cases. Using the notation above, the
Wuhan-exported cases confirmed in the 14 locations we consider
can be written as an event $(B,E,T,S) \in \mathcal{D}$ where
\begin{equation}
  \label{eq:study-sample}
  \mathcal{D} = \{(b,e,t,s) \in \mathcal{P} \mid b \le t \le e \le L,
  t \le s < \infty \}.
\end{equation}
Compared to the full population $\mathcal{P}$ of people with exposure
to Wuhan in \eqref{eq:population}, the set $\mathcal{D}$ makes three further
restrictions:
\begin{enumerate}
\item $B \le T \le E$, because we only use cases who contracted the
  virus during their stay in Wuhan;
\item $E \le L$, because the case can only be observed in the dataset
  if they left Wuhan before the travel quarantine;
\item $S < \infty$, because not all locations report asymptomatic
  cases, which motivates us to only consider COVID-19 cases who showed symptoms.
\end{enumerate}

\subsubsection*{Selection-adjusted likelihood functions}
\label{sec:select-adjust-likel}

In an ideal world where we could take independent observations
$(B_i,E_i,T_i,S_i),~i=1,\dotsc,n$ from
the exposed population $\mathcal{P}$, the likelihood function would be
given by a product of the density
$f(B_i,E_i,T_i,S_i)$ in \eqref{eq:joint-density} over $i$. However,
that is almost impossible for the initial COVID-19 outbreak in
Wuhan. Because of limited testing capacity in the beginning of the
outbreak, many COVID-19 patients in Wuhan were not identified.

Instead, \revv{in \Cref{sec:data}} we have obtained a high-quality dataset of
Wuhan-exported cases which can be considered as ``shadows'' of the
epidemic in Wuhan. To use this dataset, it is crucial that
the statistical inference takes into account the sample selection
because we do not have independent observations from
$\mathcal{P}$. Instead, we may view our sample as independent
observations generated from the following density: 
\begin{equation}
  \label{eq:conditional-density}
  f(b,e,t,s \mid \mathcal{D}) \overset{\Delta}{=} f \big(b,e,t,s
  \mid (B,E,T,S) \in \mathcal{D} \big) =
  \frac{f(b,e,t,s) \cdot 1_{\{(b,e,t,s) \in
      \mathcal{D}\}}
  }{\P\big((B,E,T,S) \in \mathcal{D}\big)},
\end{equation}
where $1_{\{\cdot\}}$ is the indicator function. \rev{To reduce cluttering,
we will omit the indicator $1_{\{(b,e,t,s) \in
  \mathcal{D}\}}$ if it is clear from the context that we are
considering sample from the Wuhan-exported cases $\mathcal{D}$.}
We can then use the product
\begin{equation}
  \label{eq:likelihood}
  \prod_{i=1}^n f \big(B_i,E_i,T_i,S_i \mid \mathcal{D}
  \big),
\end{equation}
as the likelihood function, under the assumption that we have observed
an independent and identically distributed sample
$(B_i,E_i,T_i,S_i),~i=1,\dotsc,n$ from the density
\eqref{eq:conditional-density}.

A further difficulty is that the time of transmission $T$ is usually
unobserved in our dataset. To solve this problem, we can either treat
$T$ as a latent variable and maximize the likelihood over both the
modeling parameters and the unobserved $T_i$, or simply marginalize
over $T$ in the full data likelihood and use the following observed
data likelihood,
\begin{equation}
  \label{eq:observed-data-likelihood}
  L_{\text{uncond}}(\theta)=\prod_{i=1}^n \int f
  \big(B_i,E_i,t,S_i \mid \mathcal{D} \big) \, dt,
\end{equation}
where $\theta = (f_B(\cdot), f_E(\cdot \mid \cdot), g(\cdot), h(\cdot))$ contains
all the parameters of interest.

We can also condition on $(B,E)$ to formulate a conditional likelihood function
that does not depend on the marginal distribution of $(B,E)$:
\begin{equation}
  \label{eq:observed-data-conditional-likelihood}
  L_{\text{cond}}(\theta) = \prod_{i=1}^n \int
  f_{T,S}\big(t,S_i \mid B_i,E_i, \mathcal{D}\big) \, dt,
\end{equation}
where $\theta = (g(\cdot), h(\cdot))$ and
\begin{equation}
  \label{eq:conditional-density-given-b-e}
  \begin{split}
  f_{T,S}(t,s \mid b,e, \mathcal{D}) &\overset{\Delta}{=} f_{T,S}\big(t,s \mid B=b,E=e,
  (B,E,T,S) \in \mathcal{D}\big) \\ &= \frac{f_{T,S}(t,s \mid b, e)
  }{\P((B,E,T,S) \in \mathcal{D} \mid B=b, E=e)}.
  \end{split}
\end{equation}
The information about the epidemic growth $g(\cdot)$ and the incubation
period $h(\cdot)$ contained in the density $f_{B,E}(\rev{b,e} \mid
\mathcal{D})$ is not used in the conditional likelihood, but the
benefit is that it does not require us to specify the nuisance
parameters $f_B(\cdot)$ and $f_E(\cdot \mid \cdot)$ to model the
traveling. In other words, the conditional likelihood is less
efficient than the unconditional likelihood but more robust.

Next we derive the likelihood functions
\eqref{eq:observed-data-likelihood} and
\eqref{eq:observed-data-conditional-likelihood}. For the unconditional
likelihood function we will make additional
parametric modeling assumptions on the traveling pattern $f_B(\cdot)$
and $f_E(\cdot \mid \cdot)$.

\subsubsection*{Computing the selection probability}
\label{sec:appr-select-prob}

The first technical problem is to compute
the denominators in \eqref{eq:conditional-density} and
\eqref{eq:conditional-density-given-b-e}. This is straightforward for
the conditional likelihood:

\begin{lemma} \label{lem:selection-probability-general}
  Under \Cref{assump:model-t,assump:model-s}, for $(b,e,t,s) \in \mathcal{D}$,
  \begin{equation}
    \begin{split}
    \label{eq:selection-probability-given-b-e}
    &\P((B,E,T,S) \in \mathcal{D} \mid B=b, E=e) = \nu [G(e) -
    G(b)],~\text{and} \\ &f_{T,S}(t,s \mid b,e, \mathcal{D}) = \frac{g(t)
      h(s-t)}{G(e) - G(b)},
    \end{split}
  \end{equation}
  where $G(t) = \int_{-\infty}^t g(x) \, dx$. If we additionally
  assume $g(t)$ is growing exponentially
  (\Cref{assump:model-t-parametric}), we have
  \begin{equation}
    \label{eq:conditional-density-given-b-e-parametric}
    f_{T,S}(t,s \mid b,e, \mathcal{D}) =
    \begin{dcases}
      \frac{r\exp(rt)}{\exp(re) -
        \exp(rb)} \, h(s-t),& \text{for}~r \neq 0,\\
      \frac{1}{e - b} \, h(s-t),& \text{for}~r = 0.
    \end{dcases}
  \end{equation}
\end{lemma}

An important observation here is that
\eqref{eq:conditional-density-given-b-e-parametric} does
not depend on  $\nu$ (proportion of symptomatic cases) and
$\kappa$ (absolute scale of the epidemic). Conditional likelihood
$L_{\text{cond}}(\theta)$ can then be derived by integrating
\eqref{eq:conditional-density-given-b-e-parametric} over $t$ and the
precise formula can be found in
\Cref{prop:observed-data-likelihood-parametric} below.

For the denominator in the unconditional likelihood, we need to
integrate $\P((B,E,T,S) \in \mathcal{D} \mid B=b, E=e)$ in
\Cref{lem:selection-probability-general} over the marginal
distribution of $B$ and $E$. We make the following simplifying
assumptions on $f_B(b)$ and $f_E(e \mid b)$ which heuristically say
that the travel pattern is stable during the study period:

\begin{assumption} \label{assump:model-b-parametric}
  The beginning of stay in Wuhan $B$, conditioning on $0 < B \le L$,
  follows a uniform distribution from $0$ to $L$. More specifically,
  \begin{equation}
    \label{eq:model-b-parametric}
    f_B(b) =
    \begin{cases}
      1-\pi,&~\text{for}~b=0, \\
      \pi/L,&~\text{for}~0 < b \le L,\\
    \end{cases}
  \end{equation}
  where $0 \le \pi \le 1$ is the proportion of visitors (non-residents of Wuhan)
  in the Wuhan-exposed population.
\end{assumption}

\begin{assumption} \label{assump:model-e-parametric}
  The end of stay $E$ follows an uniform distribution from $B$ to $L$
  given $E \le L$, with rate depending on whether the person resides in Wuhan:
  \begin{equation}
    \label{eq:model-e-parametric}
    \begin{split}
      &f_E(e \mid b = 0) =
      \begin{cases}
        \lambda_W,& \text{if}~0 \le e \le L, \\
        1 - L \lambda_W, & \text{if}~e = \infty,
      \end{cases},\\
      &f_E(e \mid b, b > 0) =
      \begin{cases}
        \lambda_V,& \text{if}~b \le e \le L, \\
        1 - (L-b) \lambda_V, & \text{if}~e = \infty,
      \end{cases}
    \end{split}
  \end{equation}
  where the parameters $\lambda_W,\lambda_V \le 1/L$.
\end{assumption}

For $b > 0$, \Cref{assump:model-e-parametric} implies that
$\P(E = \infty \mid b, b>0) =b \lambda_V +  (1 - L \lambda_V)$ increases
as $b$ increases. This is consistent with our intuition that the
later someone arrives in Wuhan, the more likely that person stays
there after the travel quarantine on January 23.

By using the parametric forms \eqref{eq:model-t-parametric},
\eqref{eq:model-b-parametric}, \eqref{eq:model-e-parametric} when
integrating $\P((B,E,T,S) \in \mathcal{D} \mid B=b, E=e)$ and using
the approximation $(1 + rL) / \exp(rL) \approx 0$ for $rL > 5$, we
obtain the following result.

\begin{lemma}
  \label{lem:conditional-density-parametric}
  Under
  \Cref{assump:model-t,assump:model-s,assump:model-t-parametric,assump:model-b-parametric,assump:model-e-parametric},
  for $r > 5/L$, the selection probability is given by
  \[
    \P((B,E,T,S) \in \mathcal{D}) \approx \frac{\kappa \exp(rL) \nu}{r^2}
    \Big[(1-\pi) \lambda_W + \pi \lambda_V \Big(1 - \frac{2}{rL}\Big) \Big],
  \]
  and for $(b,e,t,s) \in \mathcal{D}$, the density in \eqref{eq:conditional-density} is given by
  \begin{equation}
    \label{eq:conditional-density-parametric}
    f(b,e,t,s \mid \mathcal{D}) \approx r^2 \cdot \frac{[ 1_{\{b = 0\}} +
      (\rho/L) 1_{\{b > 0\}}] \cdot \exp(rt)}{\Big[1 + \rho (1 -
      2/(rL)) \Big] \cdot \exp(rL)} \cdot h(s-t),
  \end{equation}
  where $1_{\{\cdot\}}$ is the indicator function and $\rho =
  (\lambda_V / \lambda_W) \cdot \pi / (1-\pi)$.
\end{lemma}

Similar to \eqref{eq:conditional-density-given-b-e-parametric}, the
conditional density \eqref{eq:conditional-density-parametric} does not
depend on $\nu$ and $\kappa$. Moreover, it only depends on the
traveling parameters $\pi$, $\lambda_V$ and $\lambda_W$ through a
single transformed parameter $\rho$. The approximation $(1 + rL) /
\exp(rL) \approx 0$ is  used in the \aoasdel{Appendix}\aoasins{Online
  Supplement} to obtain the analytical formulae in
\Cref{lem:conditional-density-parametric}. It is quite reasonable for $rL
> 5$ (if the doubling time is 4 days, $rL = \log(2) /
4 \times 54 = 9.34$).

\subsubsection*{Observed data likelihood}
\label{sec:observ-data-likel}

As explained after equation \eqref{eq:likelihood}, we cannot immediately use
the densities in
\eqref{eq:conditional-density-given-b-e-parametric} and
\eqref{eq:conditional-density-parametric} for
statistical inference because we do not observe the time of
transmission $T$. The final step in the derivation of our likelihood
function is to marginalize over $t$ in the density
functions. The parametric form of $h(\cdot)$ in
\Cref{assump:model-s-parametric} allows us to derive closed-form
formulae.

\begin{proposition} \label{prop:observed-data-likelihood-parametric}
  Under
  \Cref{assump:model-t,assump:model-s,assump:model-t-parametric,assump:model-s-parametric},
  the observed data conditional likelihood
  \eqref{eq:observed-data-conditional-likelihood} 
  is given by
  \begin{equation}
    \label{eq:observed-data-conditional-likelihood-parametric}
    \begin{split}
      &L_{\text{cond}}(r,\alpha,\beta) \\
      =&
    \begin{dcases}
      r^n
      \Big(\frac{\beta}{\beta+r}\Big)^{n\alpha} \cdot \prod_{i=1}^n
      \frac{\exp(r S_i) \big[H_{\alpha,\beta+r}(S_i - B_i) -
        H_{\alpha,\beta+r}((S_i-E_i)_+) \big]}{\exp(rE_i) -
        \exp(rB_i)},& \text{for}~r>0,\\
      \prod_{i=1}^n \frac{H_{\alpha,\beta}(S_i
        - B_i ) - H_{\alpha, \beta} ((S_i - E_i)_+)}{E_i - B_i},& \text{for}~r=0,
    \end{dcases}
    \end{split}
  \end{equation}
  where $H_{\alpha,\beta }(\cdot)$ is the cumulative distribution
  function of the Gamma distribution with shape $\alpha$ and rate
  $\beta$ and $(x)_+ = \max(x,0)$ is the positive part of $x$.
  Under
  \Cref{assump:model-t,assump:model-s,assump:model-t-parametric,assump:model-s-parametric,assump:model-b-parametric,assump:model-e-parametric},
  the observed data unconditional likelihood
  \eqref{eq:observed-data-likelihood} for $r > 5/L$ is approximately given by
  \begin{equation}
    \label{eq:observed-data-likelihood-parametric}
    \begin{split}
      L_{\text{uncond}}(\rho,r,\alpha,\beta)
      \approx r^{2n} \Big(\frac{\beta}{\beta+r}\Big)^{n\alpha} \cdot \prod_{i=1}^n\bigg\{&\frac{1_{\{B_i = 0\}} +
        (\rho/L) 1_{\{B_i > 0\}} }{1 + \rho (1 - 2/(rL))} \exp\big\{r (S_i - L)\big\} \\
      & \times \big[H_{\alpha,\beta+r}(S_i - B_i) -
      H_{\alpha,\beta+r}((S_i-E_i)_+) \big]\bigg\}.
    \end{split}
  \end{equation}
\end{proposition}

It is worthwhile to point out that if $r = 0$ (the epidemic was
stationary), our conditional likelihood function
$L_{\text{cond}}(r,\alpha,\beta)$ reduces to the
likelihood function for interval-censored exposure in
\citet{reich2009estimating}. However, COVID-19 was growing quickly
during its early outbreak in Wuhan, so the growth exponent $r$ is very
different from $0$. It is thus inappropriate to use the likelihood
$L_{\text{cond}}(0,\alpha,\beta)$ to estimate the incubation period of
COVID-19, as done in some previous analyses also using
Wuhan-exported cases
\citep{backer2020incubation,lauer2020incubation,linton2020incubation}. See
\Cref{sec:estim-incub-peri} for further discussion and an illustration
of the bias due to ignoring the epidemic growth.

\subsection{Results of the parametric inference}
\label{sec:results-param-infer}

\subsubsection*{Implementation}
\label{sec:impl-likel-infer}

To fit the statistical model, we used the $378$ Wuhan-exported cases
that satisfy our sample selection criterion and do not have missing
symptom onset date. We fitted separate models for
different locations to compare the results across the locations.

As the model in \Cref{sec:model} is a regular parametric model, we
performed the usual frequentist inference using the likelihood
function \eqref{eq:observed-data-likelihood-parametric}. In
particular, point estimators of the parameters $(\rho,r,\alpha,\beta)$
were obtained by maximizing the likelihood function
\eqref{eq:observed-data-likelihood-parametric}, and confidence
intervals for the parameters were obtained by inverting the
likelihood ratio $\chi^2$-test. As we are more interested in quantiles
of the incubation period instead of the shape and rate parameters, we
parametrized the Gamma distribution in \Cref{assump:model-s-parametric}
by its median and 95\% quantile and mapped them to $\alpha$ and $\beta$
when calculating the likelihood function. The growth exponent $r$ was
also transformed to the more interpretable doubling time (in days) using
$\text{doubling time} = \log(2) / r$.

Because we only observed the date instead of the exact time for $B$,
$E$, and $S$, we applied a simple transformation before computing the
likelihood function. Instead of using the integer date which corresponds to
the end of a day, we used $B - 3/4$, $E - 1/4$, and $S-1/2$ in places of
$B$, $E$, and $S$ to compute
\eqref{eq:observed-data-conditional-likelihood-parametric} and
\eqref{eq:observed-data-likelihood-parametric}. This transformation
also avoids a singularity in the likelihood function when $B$ and $E$
are exactly equal.

\subsubsection*{Results}
\label{sec:results}

\begin{table}[t] \small 
  \centering
  \caption{Results of the parametric inference. For each location and
    parameter, the maximum likelihood estimator and the
    95\% confidence interval (in brackets) based on inverting the
    likelihood ratio test are reported.}
  \label{tab:results-parametric}
  \begin{tabular}{lccccc}
    \toprule
    \multirow{2}{*}{\bfseries Location} & {\bfseries Sample} &
                                                               \multirow{2}{*}{$\bm
                                                                                   \rho$}
    & {\bfseries Doubling time} & \multicolumn{2}{c}{\bfseries  Incubation period}
    \\
                              & {\bfseries size} & & {\bfseries (in days)} & {\bfseries
                                                            Median} &
                                                                      {\bfseries 95\% quantile} \\
    \midrule
                               \multicolumn{6}{c}{\bfseries
    Conditional likelihood} \\
    China - Hefei &  34 & Not estimated  & 2.1 (1.2--3.7) & 4.3 (2.9--6.0) &
                                                                      12.0 (9.1--17.3) \\
  China - Shaanxi &  53 & Not estimated  & 1.7 (1.0--2.8) & 4.5 (3.1--6.2) &
                                                                      14.6 (11.5--19.8) \\
  China - Shenzhen & 129 & Not estimated  & 2.2 (1.7--3.0) & 3.5 (2.8--4.3) &
                                                                       11.2 (9.5--13.6) \\
  China - Xinyang &  74 & Not estimated  & 2.3 (1.5--3.5) & 6.8 (5.4--8.2) &
                                                                      16.4 (13.8--20.1) \\
  China - Other &  42 & Not estimated  & 2.0 (1.1--3.4) & 5.1 (3.6--6.7) &
                                                                    12.3 (9.8--16.4) \\
  International &  46 & Not estimated  & 2.1 (1.4--3.4) & 3.8
                                                                (2.5--5.3)
                              & 10.9 (8.4--15.1) \\
  {\bfseries All locations} & 378 & Not estimated  & 2.1 (1.8--2.5) & 4.5 (4.0--5.0) &
                                                                    13.4 (12.2--14.8) \\
  {\bfseries All exc.\ Xinyang} & 304 & Not estimated  & 2.1 (1.7--2.5) & 4.0 (3.5--4.6)
                              & 12.2 (11.0--13.7) \\
    \midrule
                              \multicolumn{6}{c}{\bfseries
    Unconditional likelihood} \\
    China - Hefei &  34 & .40 (.18--.82) & 1.8 (1.4--2.4) & 4.1
                                                                 (2.8--
                                                                  5.5)
                              & 11.9 (9.0--17.2) \\
    China - Shaanxi &  53 & .24 (.11--.46) & 2.5 (2.0--3.1) &
                                                                    5.3
                                                                    (3.9--
                                                                    6.8)
                              & 15.0 (12.0--20.0) \\
    China - Shenzhen & 129 & .75 (.52--1.06) & 2.4 (2.1--2.8) &
                                                                     3.6
                                                                     (2.9--
                                                                     4.3)
                              & 11.3 (9.6--13.7) \\
    China - Xinyang &  74 & .45 (.27--.74) & 2.4 (2.0--2.9) &
                                                                    6.8
                                                                    (5.6--
                                                                    8.1)
                              & 16.4 (13.9--20.2) \\
    China - Other &  42 & .45 (.22--.86) & 2.1 (1.7--2.8) & 5.3
                                                                  (4.0--
                                                                  6.6)
                              & 12.4 (10.0--16.4) \\
    International &  46 & .14 (.05--.32) & 2.0 (1.6--2.6) & 3.7
                                                                  (2.5--
                                                                  5.0)
                              & 10.8 (8.4--15.1) \\
    {\bfseries All locations} & 378 & .45 (.36--.56) & 2.3 (2.1--2.5) & 4.6
                                                                  (4.1--
                                                                  5.1)
                              & 13.5 (12.3--14.9) \\
    {\bfseries All exc.\ Xinyang} & 304 & .45 (.35--.57) & 2.2 (2.1--2.5) &
                                                                       4.1
                                                                       (3.7--
                                                                       4.6)
                              & 12.3 (11.1--13.8) \\
    \bottomrule
  \end{tabular}
\end{table}

Results of the parametric model in \Cref{sec:model} are reported in
\Cref{tab:results-parametric}. We give some remarks about the
results:
\begin{enumerate}
\item There is considerable heterogeneity of the estimated $\rho$ (\aoasdel{a
  parameter capturing the traveling pattern}\aoasins{the
  ratio of outbound travel prevalences between Wuhan visitors and Wuhan
  residents}) using the unconditional
  likelihood. This is not surprising given that the locations we are
  considering are different in many ways.
\item Regardless of the location, our model shows that the epidemic
  doubling time in Wuhan was less than \aoasdel{3}\aoasins{2.5}
  days. There is no substantial heterogeneity among estimates in
  different locations.
\item The estimated incubation periods are similar for most locations
  except Xinyang, a less developed city neighboring the Hubei
  province. \aoasins{It is possible that the public health system in
    Xinyang had a lower ascertainment rate during the early outbreak.}
\item The conditional likelihood
  \eqref{eq:observed-data-conditional-likelihood-parametric} and
  unconditional likelihood
  \eqref{eq:observed-data-likelihood-parametric} give very similar
  results. Confidence intervals \revv{for the doubling time} computed using
  the unconditional likelihood are slightly shorter than those
  computed using the conditional likelihood.
\end{enumerate}

In conclusion, inferences based on our parametric model suggest that
the initial doubling time of the COVID-19 epidemic in Wuhan was
between 2 to 2.5 days, the median incubation period of COVID-19 is
around 4 days, and the 95\% quantile of the incubation period is
between 11 to 15 days.

\section{Why some previous COVID-19 analyses were severely biased}
\label{sec:why-prev-analys}

\subsection{Estimating the epidemic growth: Bias due to ignoring the travel quarantine}
\label{sec:estim-epid-growth}

In this section we discuss the selection bias in some early COVID-19
analyses. Like the present study, a highly influential article
published in the \emph{Lancet} in late January also used
Wuhan-exported cases to estimate the epidemic growth during the early outbreak
\citep{wu2020nowcasting}. However, their estimated doubling time was
6.4 days (95\% credible interval: 5.8--7.1), drastically higher than the
estimates in \Cref{tab:results-parametric}.

A closer look at the model in \citet{wu2020nowcasting} shows that the
most likely reason is that their model did not consider
how sample selection (in particular, the travel quarantine of Wuhan)
changes the likelihood function. This issue is best illustrated by
examining the marginal distribution of symptom onset in Wuhan-exported
cases, which can be obtained by integrating the conditional density
\eqref{eq:conditional-density-parametric} obtained earlier. \revv{In the
Proposition below we focus on the exported cases who are Wuhan
residents ($B=0$), whose marginal distributions of $T$ and $S$ are
slightly cleaner than those who visited Wuhan ($B > 0$).}

\begin{proposition} \label{prop:marginal-s}
  Under
  \Cref{assump:model-t,assump:model-s,assump:model-t-parametric,assump:model-b-parametric,assump:model-e-parametric},
  the marginal density of $T$ given $(B,E,T,S) \in \mathcal{D}$ for $r
  > 5/L$ 
  is approximately given by
  \begin{equation}
    \label{eq:t-marginal}
    f_T(t \mid \mathcal{D}\revv{,B=0}) \appropto \exp(rt) (L-t) \cdot 1_{\{t \le L\}},
  \end{equation}
  where $\appropto$ means approximately proportional to. If in
  addition the incubation
  period $S-T$ follows a $\text{Gamma}(\alpha,\beta)$ distribution
  (\Cref{assump:model-s-parametric}) \rev{ and $L >
    4(\alpha+5)/(\beta+r)$}, the marginal density of $S$ of
  the exported cases is approximately given by, \rev{for $s \ge L/2$},
  \begin{equation} \label{eq:s-marginal-full}
    \begin{split}
      &f_S(s \mid \mathcal{D}\revv{,B=0})\\
      \appropto& \exp(rs) \cdot \Big\{ (L-s)
    [1-H_{\alpha,\beta+r}((s-L)_+)] + \frac{\alpha}{\beta + r}[1 -
    H_{\alpha+1,\beta+r}((s-L)_+)]\Big\},
    \end{split}
  \end{equation}
  As a consequence,
  \begin{equation}
    \label{eq:s-marginal}
    f_S(s \mid \mathcal{D}\revv{,B=0}) \appropto \exp(rs) \cdot \Big(L +
    \frac{\alpha}{\beta+r} - s \Big)~\text{for}~\rev{ L/2 \le }s \le L.
  \end{equation}

\end{proposition}

\rev{The technical assumption $L >
  4(\alpha+5)/(\beta+r)$ is used to control the tail probability of
  a Gamma distribution so we may replace $H_{\alpha,\beta+r}(s)$ and
  $H_{\alpha+1,\beta+r}(s)$ by $1$ in \eqref{eq:s-marginal-full}. It
  is usually satisfied if $L$ is larger than several times
  the mean incubation period $\alpha/\beta$ and is satisfied here with the
  maximum likelihood estimator of $(r,\alpha,\beta)$ in
  \Cref{sec:results}.}

\begin{figure}[t]
  \centering
  \includegraphics[width = 0.6\textwidth]{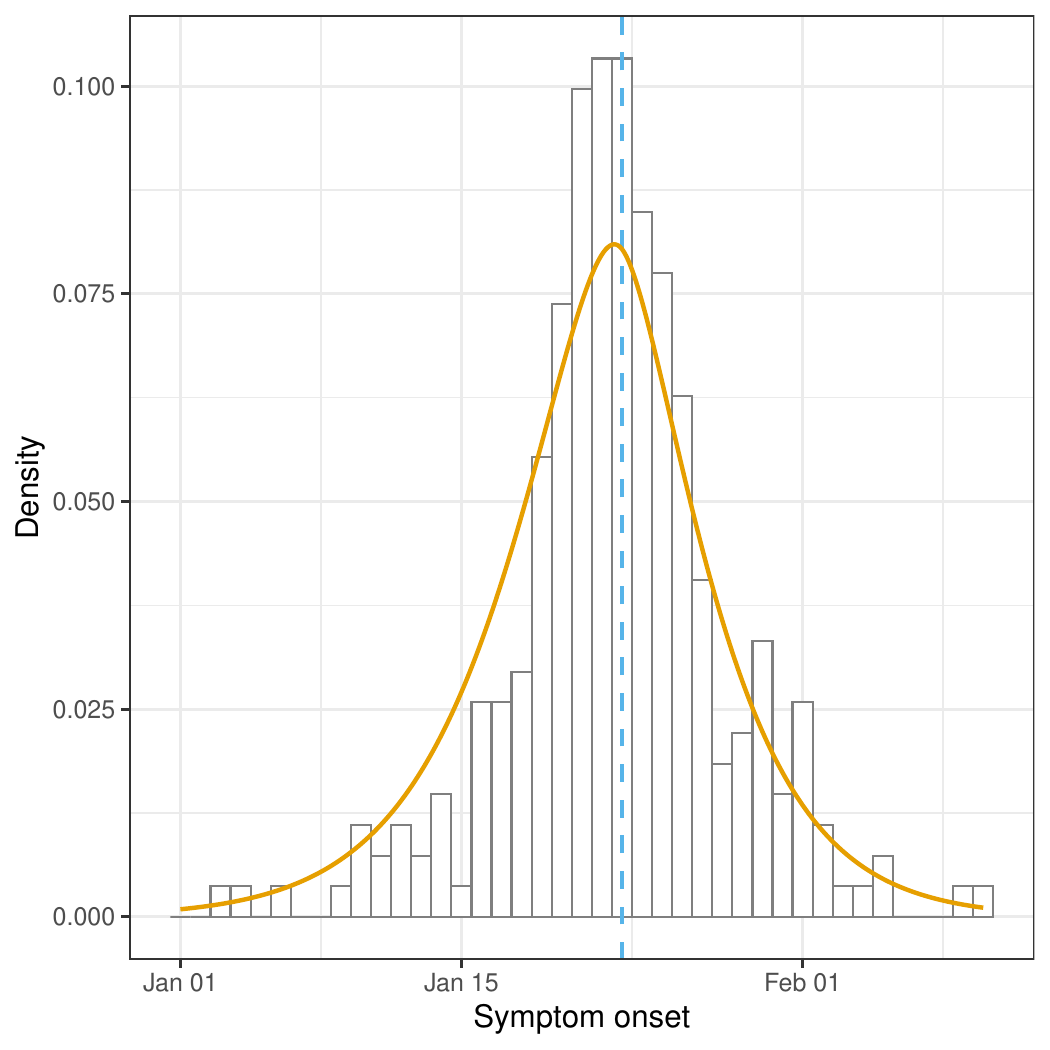}
  \caption{Marginal distribution of symptom onset of exported
    COVID-19 cases who are Wuhan residents. Histogram:
    Density of the symptom onset date of the cases dataset; Orange
    curve: Theoretical fit based on
    \eqref{eq:s-marginal-full}; Blue dashed line: Date of travel
    quarantine for Wuhan (January 23, 2020).}
  \label{fig:s-marginal}
\end{figure}

\Cref{fig:s-marginal} shows the histogram of the symptom onset of the
exported cases in our dataset who are Wuhan residents ($B=0$) and the theoretical fit based on
\eqref{eq:s-marginal-full} and the maximum unconditional likelihood
estimator in \Cref{tab:results-parametric} using all the locations ($r
= 0.30, \alpha = 1.86, \beta = 0.33$). The theoretical density provided good
fit to the observed distribution of $S$ (Pearson's $\chi^2$
goodness-of-fit test: $p$-value $=$ 0.94).

\citet{wu2020nowcasting} fitted a
Susceptible-Exposed-Infectious-Recovered (SEIR) model using
Wuhan-exported cases but did not consider sample selection due to the
travel quarantine. In the early phase of epidemic outbreaks, the SEIR
model can be well approximated by an exponential growth for cases in
Wuhan:
\[
  f_S(s) \appropto \exp(r s).
\]
However, \Cref{prop:marginal-s}, in particular equation
\eqref{eq:s-marginal}, shows that the marginal distribution of $S$ for
\revv{exported Wuhan residents $f_S(s \mid \mathcal{D},B=0)$} does not follow
the same exponential growth as $f_S(s)$.

Equation \eqref{eq:s-marginal} not only shows that fitting a simple
exponential growth to the initial symptom onsets among Wuhan-exported
cases will under-estimate the epidemic growth $r$, \revv{it can also be used}
to derive a simple bias-correction formula. \rev{We can approximate the
  log-linear regression
  for symptom incidence counts from $L - c$ to $L$ by a
  first-order Taylor expansion at the midpoint:
  \[
    \begin{split}
      \log f_S(s \mid \mathcal{D}\revv{,B=0}) &\approx rs + \log\Big( L +
      \frac{\alpha}{\beta+r} - s \Big) + \text{constant} \\
      &= rs + \log\Big( \alpha/ (\beta+r) + c/2 + (L - c/2 - s) \Big) +
      \text{constant} \\
      &\approx rs + \log\Big(\frac{\alpha}{\beta+r} + \frac{c}{2}\Big) +
      \frac{L-c/2-s}{\alpha/(\beta+r) + c/2} + \text{constant} \\
      &= \Big[r - \frac{1}{\alpha/(\beta+r) + c/2} \Big] s + \text{constant}.
    \end{split}
  \]
  Therefore the under-estimation bias is about $(\alpha/(\beta+r) +
  c/2)^{-1}$. As most of the symptom onsets of
  Wuhan-exported cases before the travel quarantine happened within two weeks, it
  might be reasonable to choose $c = 14$. Using our estimate of
  $(r,\alpha,\beta)$ in \Cref{tab:results-parametric}, the
  under-estimation bias is about $((1.86)/(0.3+0.33) + 14/2)^{-1}
  \approx 0.1$.
}

Using Wuhan-exported cases confirmed outside Mainland China by January
28, 2020, \citet{wu2020nowcasting} estimated that the doubling time of
COVID-19
was about 6.4 days, which corresponds to $r = \log(2) / 6.2 \approx
0.11$. With the above correction, \rev{the estimated $r$ would be
  $0.11 + 0.1 \approx 0.21$, or doubling time of 3.3 days. In other
  words, this simple correction already shows that the epidemic could be
  doubling twice as fast as estimated by \citet{wu2020nowcasting}.}

\rev{The actual bias of the analysis in \citet{wu2020nowcasting} is
  more complicated than the inexact calculations
  above. This is because \citet{wu2020nowcasting} fitted their
  SEIR model also using symptom onset after January 23 (time $L$). Our
  theory in \Cref{prop:marginal-s} suggests that the $f_S(s \mid
  \mathcal{D})$ not only has slower and slower growth as $s$
  approaches $L$ but also decreases eventually. This means that the
  inclusion of symptom onsets after January 23 may lead to further
  under-estimation of $r$. This also explains why, after the simple
  correction above, the epidemic growth estimate of
  \citet{wu2020nowcasting} is still not as fast as ours in
  \Cref{tab:results-parametric}.}

\subsection{Estimating the incubation period: When two biases do not
  ``balance out''}
\label{sec:estim-incub-peri}

Like the present study, several influential articles also estimated
the incubation period of COVID-19 using Wuhan-exported cases
\citep{backer2020incubation,lauer2020incubation,linton2020incubation}. Their
results are roughly in line with our estimates in
\Cref{tab:results-parametric} with lighter tails, but a closer look
shows that the existing methods actually suffer from two biases:

\begin{enumerate}
\item {\bfseries Bias due to right-truncation:} The three previous studies only
  used Wuhan-exported cases confirmed before the end of January. In
  our dataset, about 70\% of the Wuhan-exported cases were confirmed
  by that time. However, the other 30\% would have an incubation
  period of at least 8 days as they must have left Wuhan before
  January 23. The right truncation, if not accounted for, leads to
  under-estimation of the incubation period.
\item {\bfseries Bias due to ignoring epidemic growth:} The
  three previous studies all used the interval-censored likelihood
  function for the incubation period in
  \citet{reich2009estimating}. As discussed after
  \Cref{prop:observed-data-likelihood-parametric}, this likelihood
  corresponds to our conditional likelihood
  $L_{\text{cond}}(\alpha,\beta)$ with $r$ fixed at $0$ and
  thus does not account for the rapid growth of COVID-19. Intuitively,
  a person in Wuhan has a much higher prior
  probability of contracting the virus on January 20 than on January
  1, but the likelihood function in \citet{reich2009estimating} does not
  take that into account. Ignoring the epidemic growth leads to
  over-estimation of the incubation period.
\end{enumerate}

It is possible to correct for the right-truncation by further
conditioning on $S \le M$ ($M$ is some truncation time) in our
likelihood function.

\begin{proposition}
  \label{prop:likelihood-truncation}
  Under \Cref{assump:model-t,assump:model-s}, for $(b,e,t,s) \in
  \mathcal{D}$ and $s \le M$,
  \begin{equation}
    \label{eq:conditional-density-truncation}
    f_{T,S}(t,s \mid b,e,\mathcal{D}, S \le M) = \frac{g(t)
      h(s-t)}{\int_b^{\max(e,s)} g(t) H(M-t) \, dt},
  \end{equation}
  where $H(s) = \int_0^s h(x) \, dx$ is the distribution function of the
  incubation period. Furthermore, under the exponential growth model
  (\Cref{assump:model-t-parametric}) and Gamma-distributed incubation
  period (\Cref{assump:model-s-parametric}), the conditional observed
  data likelihood under the right truncation $S \le M$ is given by
  \begin{equation}\label{eq:conditional-likelihood-truncation} \small
    \begin{split}
      &L_{\text{cond,trunc}}(r,\alpha,\beta;M)\\ =&
      \begin{dcases}
        r^n \Big(\frac{\beta}{\beta+r}\Big)^{n\alpha} \prod_{i=1}^n
        \frac{\exp\{r(S_i - M)\} [H_{\alpha,\beta+r}(S_i - B_i) -
          H_{\alpha,\beta+r}((S_i - E_i)_+)]}{Z_r(M - B_i) - Z_r((M -
          E_i)_+)},& \text{if}~r \neq 0, \\
        \prod_{i=1}^n \frac{H_{\alpha,\beta}(S_i
          - B_i ) - H_{\alpha, \beta} ((S_i - E_i)_+)}{Z_0(M-B_i) - Z_0((M-E_i)_+)},& \text{if}~r=0,
      \end{dcases}
    \end{split}
  \end{equation}
  where
  \[
    Z_r(x) =
    \begin{dcases}
      \Big(\frac{\beta}{\beta+r}\Big)^{\alpha}
      H_{\alpha,\beta+r}(x) - \exp(- rx) H_{\alpha,\beta}(x),& \text{if}~r \neq 0, \\
      x H_{\alpha,\beta}(x) - \Big(\frac{\alpha}{\beta}\Big) H_{\alpha+1,\beta}(x),& \text{if}~r=0.
    \end{dcases}
  \]
\end{proposition}
It is straightforward to show that $L_{\text{cond,trunc}}(r,\alpha,\beta;M)$ reduces
to the conditional likelihood $L_{\text{cond}}(r,\alpha,\beta)$
without the right truncation in
\eqref{eq:observed-data-conditional-likelihood-parametric} when $M \to
\infty$.

We demonstrate the two kinds of biases in the estimation of the
incubation period using a retrospective experiment. In this
experiment, we assumed the incubation period follows a
Gamma distribution and estimated its median and the 95\% quantile by
maximizing one of the following three likelihood functions:
\begin{enumerate}
\item {\bfseries Adjusted for nothing:} This is the likelihood
  function in \citet{reich2009estimating} that is equal to our
  $L_{\text{cond}}(0,\alpha,\beta)$ by setting $r=0$.
\item {\bfseries Adjusted for growth:} This is our conditional
  likelihood function $L_{\text{cond}}(r,\alpha,\beta)$.
\item {\bfseries Adjusted for both growth and right-truncation:} This is our
  conditional likelihood $L_{\text{cond,trunc}}(r,\alpha,\beta;M)$
  with adjustment for sample selection due to the right-truncation $S
  \le M$.
\end{enumerate}
For each day from January 23 to
February 18, we estimated the incubation distribution using
Wuhan-exported cases in our dataset confirmed by that day. For the
third method, we choose $M$ to be a week prior to the truncation date
for confirmation, as most Wuhan-exported cases were confirmed within a
week of symptom onset. \Cref{fig:ascertainment} shows the estimated
medians and 95\% quantiles of the incubation period of COVID-19, with
pointwise confidence intervals in the plot computed using the
basic nonparametric bootstrap with 1000 resamples
\citep{efron1994introduction}.

\begin{figure}[hpt]
  \centering
  \includegraphics[width = \textwidth]{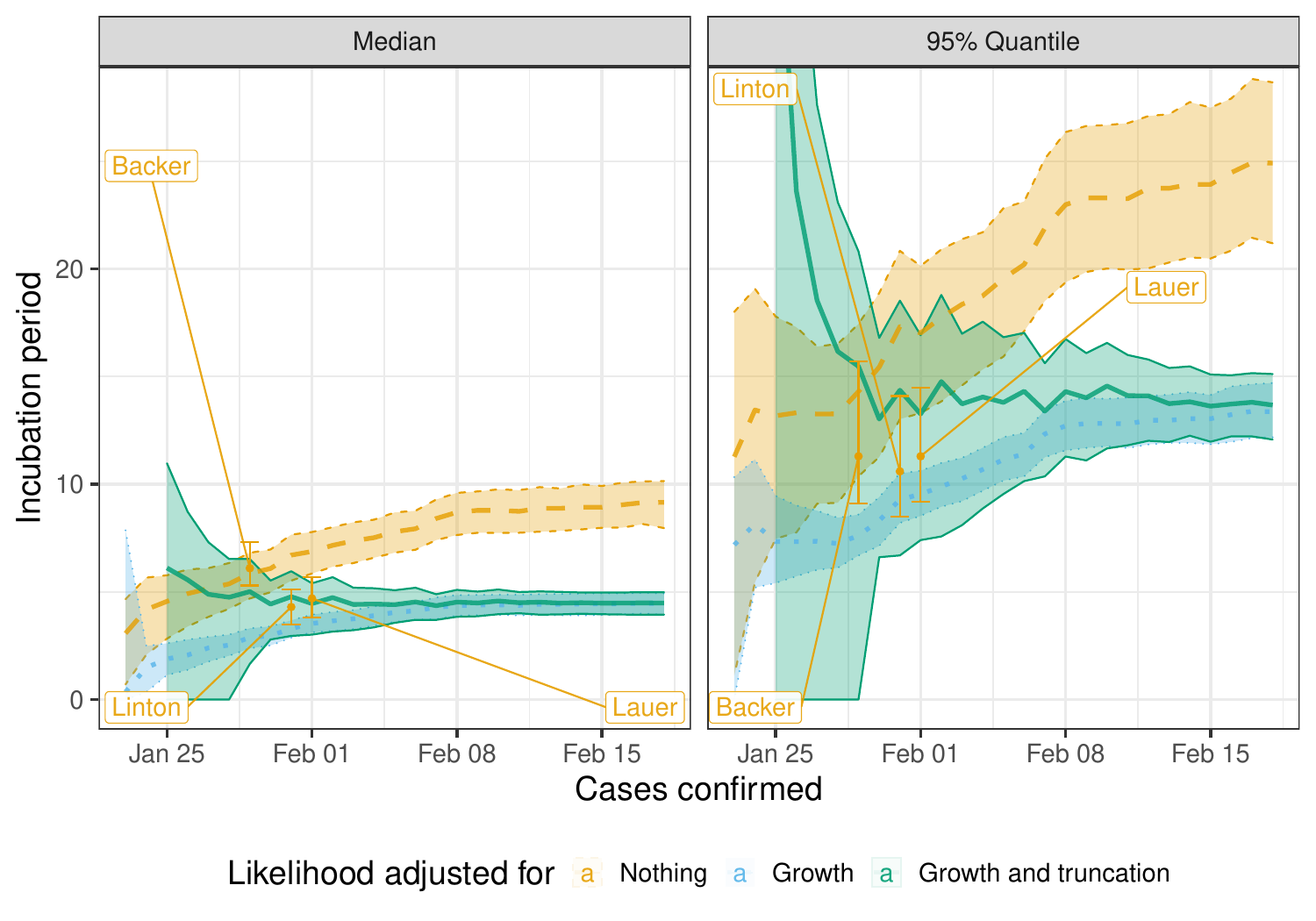}
  \caption[Illustration of bias in early studies for the incubation
  period.]{An illustration of two kinds of biases in the estimation of
    the incubation period of COVID-19. The curves (region) in the plot are
    maximum likelihood estimators (and bootstrap confidence intervals)
    using three likelihood functions and cases confirmed by each
    day. Likelihood functions used in this experiment are:
    $L_{\text{cond}}(0, \alpha,\beta)$ (dashed orange),
    $L_{\text{cond}}(r,\alpha,\beta)$ (dotted blue), and
    $L_{\text{cond,trunc}}(r,\alpha,\beta;M)$ (solid green). Results
    of some previous studies \citep{backer2020incubation,lauer2020incubation,linton2020incubation}
    (essentially using our conditional likelihood with $r$ set to 0 on
    different datasets) are also shown in this
    plot.\\
    \small (Note: \citet{lauer2020incubation}
    did not report an estimated 95\% quantile of the incubation period. Here we imputed
    it based on the reported median and 97.5\% quantile, assuming a
    Gamma distribution for the incubation period. Although
    \citet{lauer2020incubation} used COVID-19 cases confirmed as late as
    late February, only $4$ out of their 181 cases were confirmed in
    February. In this Figure it is thus treated as using cases confirmed
    up till February 1.)}
  \label{fig:ascertainment}
\end{figure}

The bias due to not accounting for right-truncation can be clearly
visualized from the dotted blue curves in \Cref{fig:ascertainment}. Had
we fitted our conditional likelihood function
$L_{\text{cond}}(r,\alpha,\beta)$ using cases confirmed by January 31
(265 cases),
the estimated median incubation period would be 3.5 days and the 95\%
quantile would be 9.5 days. In comparison, when the entire dataset is
used, the estimated median and 95\% quantile are 4.6 days and 13.5
days (\Cref{tab:results-parametric}).

The over-estimation due to ignoring the epidemic growth is even more
dramatic. Had we fitted
the incubation period using the likelihood
function in \citet{reich2009estimating} (the same as setting $r=0$ in
our conditional likelihood) to all the cases in our dataset (387 cases), the
estimated median incubation period would be 9.2 days and the 95\%
quantile would be a whopping 24.9 days!

The truncation-corrected conditional likelihood
$L_{\text{cond,trunc}}(r,\alpha,\beta;M)$ derived in
\Cref{prop:likelihood-truncation} successfully corrected for the
right-truncation bias. The estimated median and 95\% quantile of the
incubation using $L_{\text{cond,trunc}}(r,\alpha,\beta;M)$ were
roughly unbiased starting from the end of January. Had we fitted this
likelihood using all cases confirmed by January 31 and having shown symptoms
a week prior (220 cases), the estimated median incubation period would
be 4.8 days (95\% CI: 3.0 to 6.0) and the estimated 95\% quantile
would be 14.4 days (95\% CI: 6.7 to 18.5). These estimates are less
precise than the estimates obtained using the entire dataset
(\Cref{tab:results-parametric}), but they correctly reflect the
uncertainty due to the right-truncation. In contrast, using the
wrong likelihood functions not only results in biased point estimates
but also narrow and misleading confidence intervals.

Because the right-truncation bias and epidemic growth bias are
towards opposite directions, coincidentally
they were almost ``balanced out'' in the previous studies. As a
consequence, their estimates were not drastically different from
ours. To be fair in our criticism, the previous studies did
acknowledge that under-ascertainment of mild cases could bias their
analyses. \citet{backer2020incubation} mentioned the over-estimation
due to ignoring epidemic growth in their
discussion. \citet{linton2020incubation} attempted to use a formula to
correct for right-truncation which bears some similarity to
\eqref{eq:conditional-density-truncation}, which resulted in slightly
longer estimates of the incubation period. However,
\citet{linton2020incubation} did not give any justification to the
formula and we could not derive it from our generative
model. In any case, our experiments in \Cref{fig:ascertainment}
clearly show that these early estimates of the incubation period
(especially their tail estimates) are unreliable to guide health
policies.


\section{Nonparametric inference}
\label{sec:bayes}


So far we have used parametric assumptions (e.g.\ Gamma-distributed
incubation period) to explicitly derive likelihood functions for the
observed data. To assess the robustness of our results, we next
relax some of these parametric assumptions. In particular, we will
model the distribution of the incubation period nonparametrically so
\aoasins{there are fewer restrictions on}
the tail probabilities\aoasdel{are not determined by any parametric
form}. \aoasins{By a nonparametric model, we mean the inference is not
based on any particular parametric family of probability
distributions. This is slightly different from the standard notion of a
``Bayesian nonparametric model'' that involves modeling an
infinite-dimensional parameter (e.g. a function on the real line).}
Because analytic forms of the sample selection probabilities
$\P((b,e,t,s) \in \mathcal{D}  \mid b, e)$ and $\P((b,e,t,s) \in
\mathcal{D})$ are generally unavailable, we will put prior
distributions on the model parameters and use Markov Chain Monte Carlo
(MCMC) to compute their posterior distributions.

\subsection{Time discretization}
\label{sec:time-discretization}

We start by discretizing all the time variables in the model, which are measured
in days. This will simplify the Bayesian computation.
Instead of working with continuous time $(B,E,T,S) \in
\mathcal{P}$, we use the discretization:
\[
  B^{*} = \lceil B \rceil, E^{*} = \lceil E \rceil, T^{*} = \lceil T
  \rceil, S^{*} = \lceil S  \rceil,
\]
where $\lceil \cdot \rceil$ is the ceiling function ($\lceil x \rceil$
is the smallest integer larger than $x$). The support of
$(B^{*},E^{*},T^{*},S^{*})$ is then $\mathcal{P}$, the set of all $4$-tuples
of integers and $\infty$. The general continuous distributions in
\Cref{assump:model-t,assump:model-s} can be modified accordingly:
\[
  \begin{split}
    \P(T^{*} = t^{*} \mid B^{*} = b^{*}, E^{*} = e^{*}) &=
    \begin{cases}
      g^{*}(t^{*}),& \text{if}~b^{*} \le t^{*} \le e^{*},\\
      1 - \sum_{t^{*} = b^{*}}^{e^{*}} g^{*}(t^{*}), & \text{if}~t^{*} = \infty;
    \end{cases} \\
    \P(S^{*} = s^{*} \mid B^{*} = b^{*}, E^{*} = e^{*}, T^{*} = t^{*})
    &= \begin{cases}
      \nu \cdot h^{*}(s^{*}-t^{*}),& \text{if}~t^{*} \le s^{*} < \infty,\\
      1 - \nu, & \text{if}~s^{*} = \infty,
    \end{cases} \\
  \end{split}
\]
where $g^{*}(\cdot)$ satisfies $\sum_{x^{*}=0}^{L} g^{*}(x^{*}) \le 1$
and $h^{*}(\cdot)$ is a probability mass function on nonnegative
integers: $\sum_{x^{*}=0}^{\infty} h^{*}(x^{*}) = 1$.

\subsection{Relaxing the parametric assumptions}
\label{sec:relax-param-assumpt}

Our parametric assumptions
(\Cref{assump:model-t-parametric,assump:model-s-parametric,assump:model-b-parametric,assump:model-e-parametric})
on the distribution of $(B,E,T,S)$ can be translated to the following
assumptions on $(B^{*},E^{*},T^{*},S^{*})$ after discretization:
\[
  g^{*}(t^{*}) \approx g_{\kappa,r}(t^{*}) = \kappa \exp(r t^{*}),~h^{*}(t^{*} - s^{*}) \approx
  h_{\alpha,\beta}(t^{*} - s^{*}),
\]
\[
  \P(B^{*} = b^{*}) =
  \begin{cases}
    (1 - \pi),& \text{for}~b^{*} = 0,\\
    \pi / L,& \text{for}~b^{*} = 1,\dotsc,L,
  \end{cases}
\]
and
\[
  \P(E^{*} = e^{*} \mid B^{*} = b^{*}) =
  \begin{cases}
    \lambda_{b^{*}},& \text{for}~b^{*} \le e^{*} \le L, \\
    1 - (L - b^{*}+1) \lambda_{b^{*}},& \text{for}~ e^{*} = \infty.
  \end{cases}
\]
where $\lambda_0 = \lambda_W$ and $\lambda_1 = \cdots = \lambda_L =
\lambda_V$.

In the nonparametric model we consider the following relaxations:
\begin{enumerate}
\item \textbf{Nonparametric distribution for the incubation period:}
  Besides putting a prior to encourage smoothness and log-concavity, we
  do not put any parametric restrictions on the distribution of the
  incubation period.
\item \textbf{Two-stage exponential growth:} Human-to-human
  transmissibility of COVID-19 is first confirmed to the public in the
  evening of January 20. We modify the exponential growth model to
  allow for a different growth exponent after January 20:
  \[
    g^{*}(t^{*}) = g_{\kappa,r_1,r_2}^{*}(t^{*}) =
    \begin{cases}
      \kappa \exp(r_1 t^{*}) &\text{if}~t\le L_1, \\
      \kappa \exp(r_2 (t^{*} - L_1) + r_1^{*} L_1) & \text{if}~L_1 < t \le L_2,
    \end{cases}
  \]
  where $L_1 = 51$ (January 20) and $L_2 = L = 54$ (January 23). The
  simple exponential growth model is a special case of this model with
  both $L_1$ and $L_2$ set to $L$.
\item \textbf{Geometric distribution for $E^{*} \mid B^{*}$:} As a
  sensitivity analysis to our assumption that $E^{*} \mid B^{*}$ is
  uniformly distributed between $B^{*}$ and $L$, this relaxation
  assumes a geometric distribution for $E^{*} \mid B^{*}$:
  \[
    \P(E^{*} = e^{*} \mid
    E^{*} \ge e^{*}, B^{*}) =
    \begin{cases}
      \eta_{B^{*},1} & \text{ if } e^{*} < L_\text{chunyun}, \\
      \eta_{B^{*},2} & \text{ if } e^{*} \ge L_\text{chunyun},
    \end{cases}
  \]
  where $L_\text{chunyun}=41$ corresponds to January 10, the start of
  the Chinese New Year travel season known as ``chunyun''. We assume
  $\eta_{0,i}=\eta_{W,i}$ and $\eta_{1,i}=\dots=\eta_{L,i}=\eta_{V,i}$,
  for $i=1,2$.

\item \revv{\textbf{Gender-specific and age-specific incubation periods}: To assess
    whether the distribution of incubation period varies with gender,
    we use different densities, $h^{*}_M(\cdot)$ for men and $h^{*}_F(\cdot)$ for
    women. Like in (i), we put no parametric restrictions on these
    distributions apart from the same prior that encourages smoothness
    and log-concavity. Similarly, we can use different densities for
    different age groups. To avoid fitting incubation period with too
    few observations, we only consider two age groups: above 50 years
    old and below 50 years old. Notice that the same exponential
    growth model for $g$ is used for different gender or age groups,
    as we expect the growth of the chance of infection is the same for
    all strata.}
\end{enumerate}

Under these different modeling assumptions, likelihood functions for
the parameters can be computed in the same way as in
\Cref{sec:sample-selection}, with integrals replaced by finite
sums. We omit the details here.

\subsection{Prior distributions and details of the implementation}
\label{sec:prior-distr-deta}

To simplify the computation, we assume the incubation period of
COVID-19 is less than 30 days. It is common to use a unimodal
distribution with a smooth density function to model the incubation
period. We use the following prior distribution on $h^{*}(\cdot)$ to
encourage smoothness and log-concavity:
\begin{equation}
  \label{eq:prior-h}
  \begin{split}
    &\pi\big(h^{*}(0),\dotsc,h^{*}(29)\big) \propto \Big(
    \prod_{x^{*}=0}^{29} h^{*}(x^{*})^{\mu \cdot h_0(x^{*}) - 1}
    \Big) \\
    &\qquad \qquad\times \exp\Big\{\sum_{x^{*}=1}^{28} \big( 2\log h^{*}(x^{*}) -
    \log h^{*}(x^{*} - 1) - \log h^{*}(x^{*}+1)  \big)_{-}
    \Big\}.
  \end{split}
\end{equation}
where $(\cdot)_-$ is the negative part function. The first part of
the right hand side of \eqref{eq:prior-h} is proportional to the
density of a Dirichlet distribution with concentration parameters
$\{ \mu \cdot h_0(0),\dotsc,\mu \cdot h_0(29) \}$. We choose $h_0(\cdot)$
to be a discretization of $\text{Gamma}(9, 1.5)$, whose tail probability of
$\ge 14$ days is less than $0.01$. The second part of the right hand side of
\eqref{eq:prior-h} is an exponential tilt which penalizes lack of
log-concavity.

We put uninformative priors on other parameters in the model:
\[
  r_1 \sim \text{Exp}(1),~r_2 \sim \text{N}(0, 4),\kappa \sim
  \text{Unif}(0,1),~\lambda_W,\lambda_V \sim \text{Unif}(0, 1/L).
\]
Note that $r_2$ is allowed to be negative (exponential decrease after
January 20). For the model with a geometric distribution for $E^{*}\mid B^{*}$,
we put Unif$(0,1)$ priors on $\eta_{W,1},\eta_{W,2},\eta_{V,1},\eta_{V,2}$.

A random walk Metropolis--Hastings algorithm targeting the posterior distribution
of $h^{*}(\cdot)$ and $r_1$, $r_2$ was implemented using the TensorFlow Probability
library in Python~\citep{dillon2017tensorflow}. We simulated chains of $80,000$ steps,
discarding a burn-in period of 50\%. The convergence of the sampler was assessed
by simulating 8 parallel Markov chains with initial values overdispersed with respect to the
target distribution, and computing the potential of scale reduction
factor \citep{gelman1992inference} for $r_1$, the mean of the
incubation period distribution, and the probability of an incubation
period of 14 days or more. In every case, the statistic was
confidently below $1.1$.


\subsection{Results of Bayesian nonparametric inference}
\label{sec:bayes_result}

\subsubsection*{Aggregated results}

\Cref{tab:results-bayesian} reports the results of the Bayesian
nonparametric inference in $7$ different scenarios. Overall, they are
not too dissimilar to the results of the parametric model in
\Cref{tab:results-parametric}. Without restricting the tail to follow
that of a Gamma distribution, the estimated tail probabilities are
slightly higher than those in \Cref{tab:results-parametric}. The
posterior mean for $\P(S^{*}-T^{*} \ge 14~\text{days})$ exceeds $0.03$
in all scenarios, even when we exclude the cases
confirmed in Xinyang who seemed to have longer incubation periods in
\Cref{tab:results-parametric}.

Moreover, prior and posterior
distributions of $\P(S^{*}-T^{*} \ge 14~\text{days})$ show a large
discrepancy (\Cref{fig:results-bayesian-2}), indicating that the
posterior estimates of the tail probabilities are driven by
the data instead of the prior. \revv{Employing the two-stage epidemic
  growth model suggests that the epidemic growth may have slowed down
  after January 20, but this more flexible model did not alter the
  estimated doubling time and incubation period distribution
  substantially.} Taken together, our nonparametric models suggest
that the probability of an incubation period of at least 14 days
(among symptomatic cases) may be about $5\%$.


\begin{figure}[ht]
  \centering
  \begin{subfigure}[b]{0.65\textwidth}
    \includegraphics[width=0.9\textwidth]{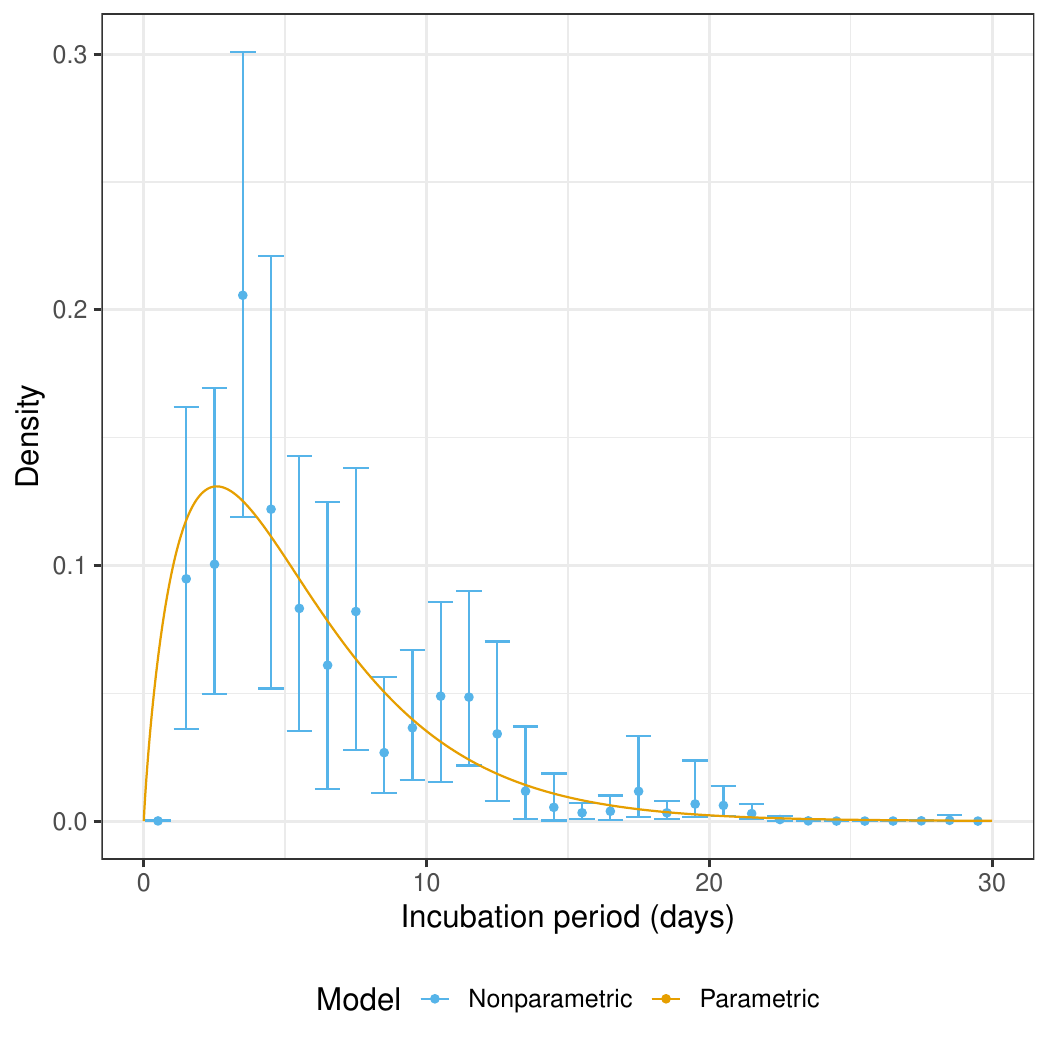}
    \caption{A comparison of the estimated parametric and
      nonparametric incubation period distributions. Orange curve:
      Gamma distribution with median = 4.5 and 95\% quantile =
      13.4. Blue error bars: Posterior mean and 95\% credible
      interval, computed in the model with a two-stage exponential growth,
      and geometric $E^*\mid B^*$, with $\mu=10$.}
    \label{fig:results-bayesian-1}
  \end{subfigure} \\[4mm]
  \begin{subfigure}[b]{0.9\textwidth}
    \includegraphics[width=\textwidth]{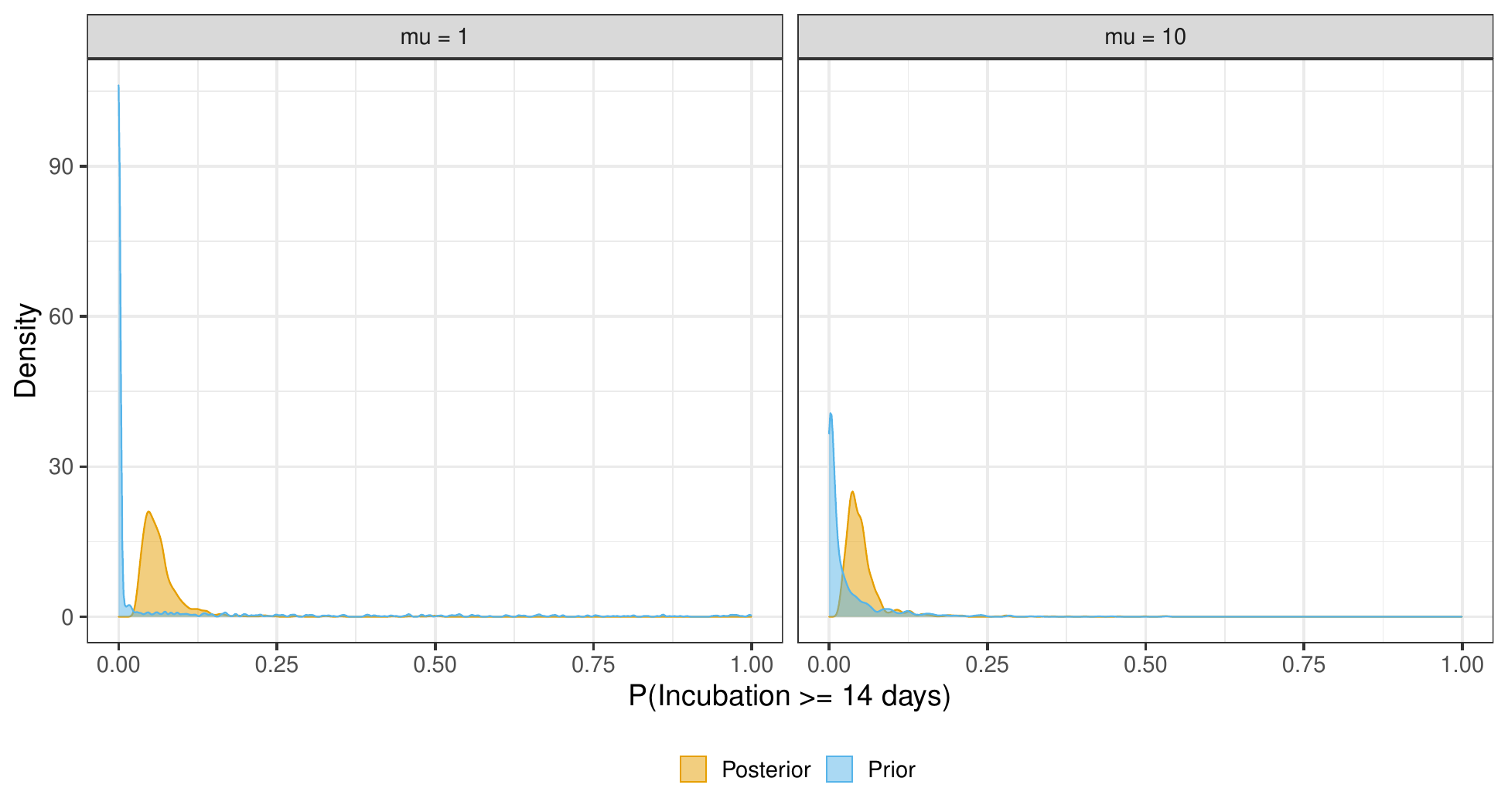}
    \caption{A comparison of the prior and posterior
      $\P(\text{Incubation} \ge 14~\text{days})$ in the first two simulation
      scenarios of \Cref{tab:results-bayesian}. Two panels
      corresponds to two choices for the prior distribution \eqref{eq:prior-h} of
      $h^{*}(\cdot)$.}
    \label{fig:results-bayesian-2}
  \end{subfigure}
  \caption{An illustration of the nonparametric Bayesian fit to the
    incubation period distribution.}
  \label{fig:results-bayesian}
\end{figure}

\revv{
\subsubsection*{Gender-specific and age-specific incubation periods}

\begin{figure}[tp]
  \centering
  \begin{subfigure}[b]{\textwidth}
    \includegraphics[width =
    0.45\textwidth]{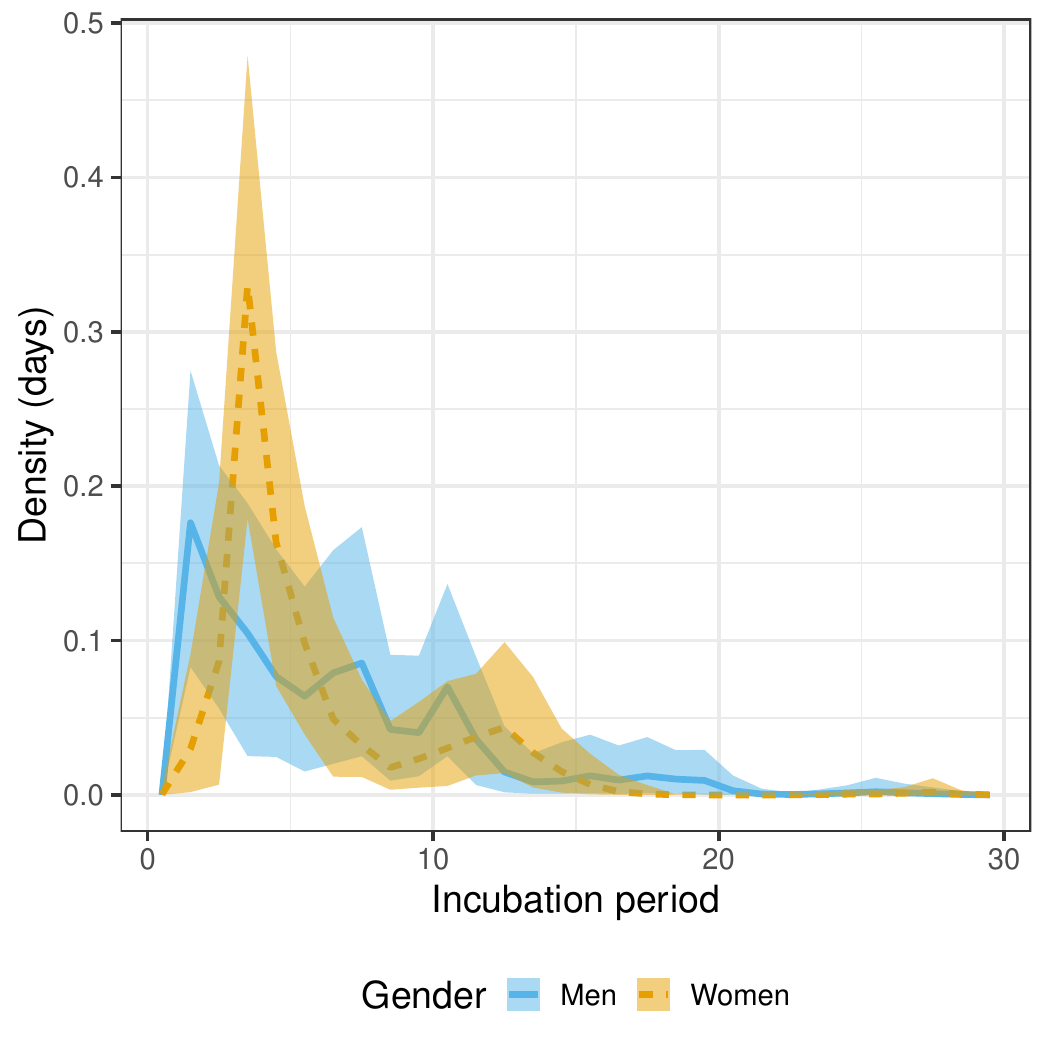} \qquad
    \includegraphics[width =
    0.45\textwidth]{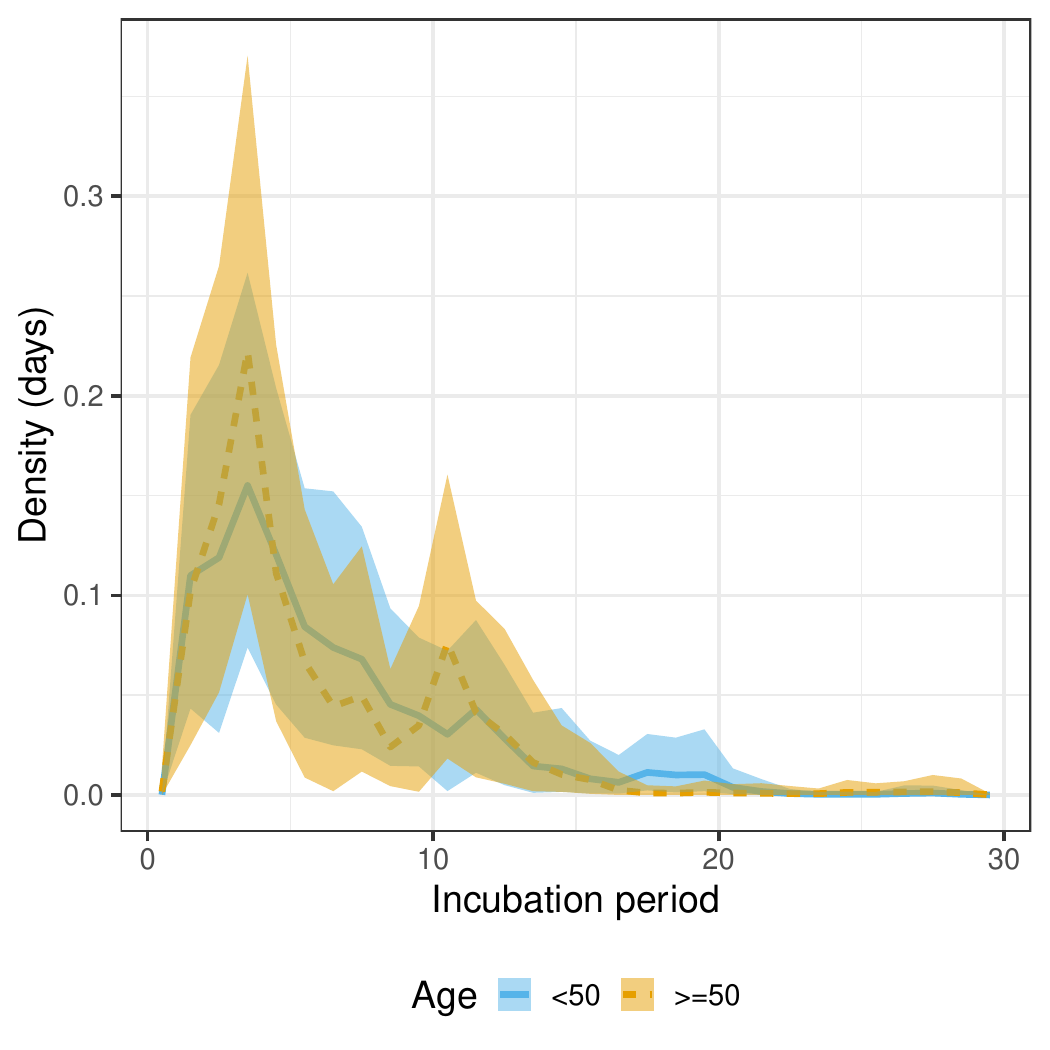}
    \caption{Gender-specific (left) and age-specific (right)
      distributions of the incubation period.}
    \label{fig:specific-incubation}
  \end{subfigure}
  \begin{subfigure}[b]{\textwidth}
    \includegraphics[width =
    0.45\textwidth]{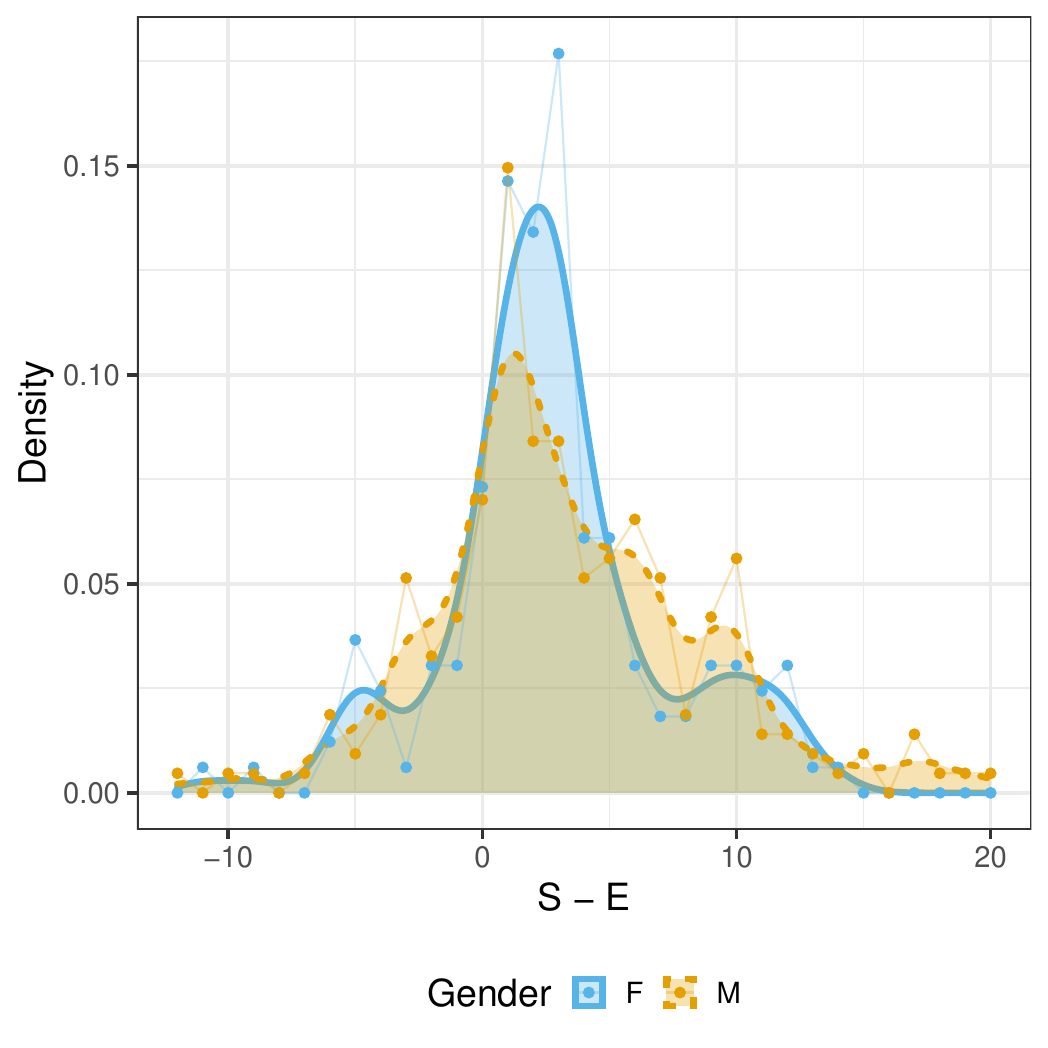} \qquad
    \includegraphics[width =
    0.45\textwidth]{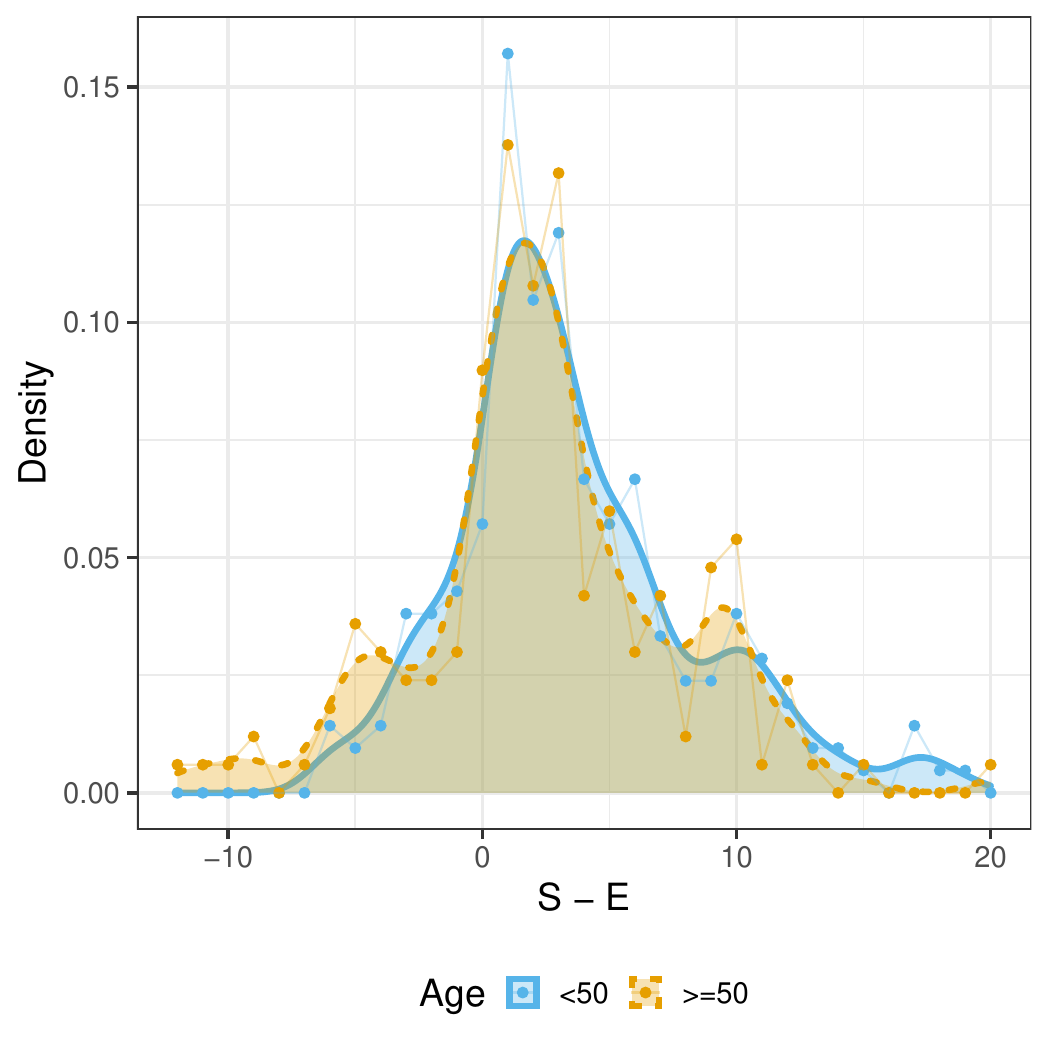}
    \caption{Gender-specific (left) and age-specific (right)
      distributions of $S - E$ (days from leaving Wuhan to symptom
      onset). Dots represent the exact proportions in the dataset and
      the curves are kernel density estimates with bandwidth set to 1
      day. }
    \label{fig:specific-s-e}
  \end{subfigure}
  \caption{Results of the Bayesian nonparametric inference after
    stratification by gender and age. Blue solid curves are for
    men and cases younger than 50 and orange dashed curves are for
    women and cases older than 50.}
  \label{fig:stratification}
\end{figure}

The gender-specific and age-specific estimates of the incubation
period can be found in
\Cref{tab:specific-incubation} and \Cref{fig:specific-incubation}. The
estimated distributions of incubation periods for men and women are notably
different, with the distribution of men peaking earlier and having a
heavier tail than women. In particular, men seem to be much more
likely than women to develop symptoms within two days of infection. A
related phenomenon can be directly seen
from the raw distribution of $S - E$ (days from leaving Wuhan to symptom
onset; can be negative if a person showed symptoms during the stay in
Wuhan) for the cases exported from Wuhan (\Cref{fig:specific-s-e}). The
difference $S-E$ appears to be more spread out in men. A
nonparametric Ansari-Bradley test for the dispersion
gives a $p$-value $\approx 0.0025$, while the Wilcoxon signed-rank
test for the location gives a $p$-value $\approx 0.48$
\citep{sidak1999theory}.

The difference of incubation period distributions for older and
younger cases seem less pronounced in
\Cref{fig:specific-incubation}. Their distributions of $S -E$ also
appear to be quite similar in \Cref{fig:specific-s-e} (Ansari-Bradley
test for dispersion: $p$-value $\approx 0.56$; Wilcoxon test for
location: $p$-value $\approx 0.1$).

}

\begin{landscape}
    \begin{table}[ht] \small 
      \centering
      \caption{Results of the nonparametric Bayesian
        inference where we do not impose a parametric form for the
        distribution of the incubation period. As sensitivity
        analyses, we also vary the study sample, model for the
        epidemic growth, distribution of $E^{*}$ given $B^{*}$, and
        the hyperprior parameter $\mu$. Numbers reported in the table
        are posterior means and 95\%  credible intervals (in
        brackets).}
      \label{tab:results-bayesian}
      \begin{tabular}{ll|ccccccc}
        \toprule
        \multicolumn{2}{c|}{Sample} & All & All &
                                                  Shenzhen
        & Wuhan residents & All except Xinyang & All & All\\
        \multicolumn{2}{c|}{Growth} & $r_1$ & $r_1$ & $r_1$ & $r_1$ &
                                                                      $r_1$
                                               & $r_1,r_2$ & $r_1,r_2$ \\
        \multicolumn{2}{c|}{$E^{*} \mid B^{*}$} & Uniform & Uniform & Uniform & Uniform
                          & Uniform & Uniform & Geometric \\
        \multicolumn{2}{c|}{$\mu$} &
                                                                    $\mu
                                                                    = 1$ &
                                                                           $\mu
                                                                           = 10$
                                                &
                                                  $\mu
                                                  =
                                                  1$
        & $\mu = 1$
                          & $\mu = 10$ & $\mu =
                                            1$ &
                                                 $\mu
                                                 = 1$ \\
        \midrule
        \multicolumn{2}{c|}{Doubling days for $r_1$} &
 2.4 (2.2--2.6) & 2.4 (2.2--2.6) & 2.5 (2.2--2.9) & 2.3 (2.1--2.6) & 2.4 (2.2--2.6) & 2.2 (2.0--2.4) & 2.2 (1.9--2.4)
 \\

\multicolumn{2}{c|}{$r_2$ (Growth in Jan.\ 21--23)}  &
-- & -- & -- & -- & -- & .01 (-.23--.22) & -.12 (-.40--.12)
 \\
        \midrule
        \multirow{5}{*}{\begin{tabular}{c} Incubation \\
                          period \end{tabular}} & Mean &
5.5 (5.0--5.9) & 5.4 (5.0--5.9) & 4.4 (3.7--5.1) & 5.6 (5.0--6.1) & 5.0 (4.6--5.5) & 5.6 (5.1--6.1) & 5.6 (5.2--6.1)
 \\

    & $\P(\ge 7)$ &
.31 (.25--.38) & .31 (.25--.37) & .22 (.13--.31) & .30 (.22--.38) & .26 (.19--.32) & .05 (.02--.08) & .06 (.03--.08)
 \\

    & $\P(\ge 10)$ &
.19 (.14--.24) & .18 (.14--.22) & .10 (.05--.17) & .21 (.14--.28) & .14 (.10--.19) & .19 (.14--.23) & .19 (.15--.24)
 \\

    & $\P(\ge 14)$ &
.05 (.02--.08) & .04 (.02--.07) & .03 (.01--.06) & .04 (.01--.07) & .03 (.01--.06) & .05 (.02--.08) & .06 (.03--.08)
 \\

    & $\P(\ge 21)$ &
.00 (.00--.01) & .00 (.00--.00) & .01 (.00--.03) & .01 (.00--.02) & .00 (.00--.01) & .00 (.00--.01) & .00 (.00--.01)
 \\
        \bottomrule
      \end{tabular}
    \end{table}

  \begin{table}[h] \small 
  \centering
  \caption{Estimated distributions of incubation period in
    different subgroups.}
  \label{tab:specific-incubation}
  \begin{tabular}{l|ccc|ccc}
    \toprule
    Subgroup & Men & Women & Difference & Age $<$ 50 & Age $\ge$ 50 &
                                                                       Difference \\
    \midrule
Mean &
5.8 (5.1--6.6) & 5.3 (4.7--5.9) & 0.5 (-0.4--1.5) & 5.4 (4.7--6.1) & 5.8 (5.2--6.5) & -0.5 (-1.4--0.4)
 \\

$\P(\ge 2)$ &
.82 (.73--.92) & .97 (.91--1.00) & -.15 (-.25---.04) & .90 (.78--.97) & .89 (.81--.96) & .01 (-.11--.13)
 \\

$\P(\ge 4)$ &
.59 (.50--.69) & .55 (.42--.69) & .04 (-.13--.20) & .53 (.41--.65) & .62 (.50--.72) & -.09 (-.25--.07)
 \\

$\P(\ge 7)$ &
.37 (.29--.46) & .24 (.17--.32) & .13 (.01--.24) & .31 (.22--.40) & .34 (.26--.42) & -.03 (-.14--.09)
 \\

$\P(\ge 10)$ &
.20 (.13--.28) & .17 (.11--.23) & .03 (-.06--.14) & .20 (.13--.28) & .18 (.13--.24) & .01 (-.08--.11)
 \\

$\P(\ge 14)$ &
.07 (.04--.12) & .03 (.01--.07) & .04 (-.01--.09) & .03 (.01--.07) & .07 (.04--.11) & -.03 (-.08--.01)
 \\

$\P(\ge 21)$ &
.01 (.00--.03) & .00 (.00--.02) & .00 (-.01--.02) & .01 (.00--.03) & .01 (.00--.01) & .00 (-.01--.02)
 \\
    \bottomrule
  \end{tabular}
\end{table}


\end{landscape}




\section{Discussion}
\label{sec:discussion}

In this article, we have proposed the generative BETS model for
four key epidemiological events: beginning of exposure, end of
exposure, time of transmission, and time of symptom onset. Under
parametric models, we have derived the sample inclusion probability for
exported cases and used it to correct for selection bias in the
likelihood functions. Across different sub-samples and modeling
assumptions, the initial epidemic doubling time for COVID-19 in Wuhan
is consistently estimated to be between $2$ to $2.5$ days. Our
nonparametric Bayesian analysis suggests that the parametric fit
likely under-estimated the tail of
the incubation period, and among all the COVID-19 patients who
develop symptoms, about 5\% of them could develop the symptoms at least 14
days after contracting the pathogenic virus. \revv{Gender-specific
  analysis shows that men may have a more variable incubation
  period than women. In particular, more men appear to show symptoms within two
  days of infection, which could be related to the men's higher death
  rate across the world \citep{globalhealth5050}.} We hope the
generality of our model makes it extensible in further studies of
the current pandemic and other outbreaks in the future.



A key epidemiological parameter we decided not to study in this
article is the basic reproduction number, commonly denoted by
$R_0$. Intuitively, $R_0$ is the expected number of secondary
infections produced by a
typical case in a population where everyone is susceptible. In early
outbreak analysis, $R_0$ can be estimated from the epidemic growth
exponent $r$ by $R_0 = 1/M(-r)$ \citep{wallinga2007generation}, where
$M(\cdot)$ is the moment generating function for the distribution of
the serial interval (time between successive cases in a chain of
transmission). Several studies have attempted to estimate the serial
interval of COVID-19 in Wuhan by using observed pairs of infector-infectees
\citep{li2020early,nishiura2020serial,du2020serial}. The reported
point estimate of the mean serial interval ranging from $4.0$
\citep{du2020serial} to $7.5$ days \citep{li2020early}. However, for most
COVID-19 cases it seems impossible to ascertain the infector, so these
early estimates of the serial interval could be severely biased by
sample selection just like the early estimates of epidemic growth and
incubation period as seen in \Cref{sec:why-prev-analys}.

Our findings in this article should be viewed together with
the limitations of our methodology. \revv{First of all, symptom onset time were usually
  reported by patients, who could be under social pressure to
  report a later symptom onset (for example, so they did not travel
  when showing symptoms). This can make estimated incubation longer
  than the truth.} Second, although the contact tracing for
travelers from Wuhan was intensive in the locations included in our
dataset, some degree of under-ascertainment of Wuhan-exported
cases is perhaps inevitable. \revv{If patients who showed symptoms
  earlier were less likely to be ascertained, our analysis may have
  over-estimated the speed of the epidemic growth.} \revv{Third, the discernment
of Wuhan-exported cases in \Cref{sec:wuhan-exported-cases} is not
perfect; for example,} there is
ambiguity about where some COVID-19 cases were infected if they both
had stayed in Wuhan and were exposed to other confirmed cases after
their stay. Another potential limitation is the core assumptions that
the disease transmission and progression are independent of
traveling. This assumption is necessary to
extend the conclusions from a ``shadow'' of the epidemic
(Wuhan-exported cases) to the center of the outbreak, but it can be
violated if, for example, some people canceled travel plans due to
feeling sick. Finally, it is possible that the population of travelers
is not representative of the general population in a meaningful
way.

\revv{Nevertheless, these limitations are perhaps minor compared to
the selection bias identified in this article. Several authors have
warned about selection bias and other statistical issues in COVID-19
studies
\citep{lipsitch2020estimating,lipsitch2020use,niehus2020quantifying,gelman2020blog,bergstrom2020interview}. By
constructing a generative model and deriving the likelihood functions
from first principles, we gave a quantitative assessment of the selection bias in several
high-impact studies. We found that the biases were indeed startling. This
highlights the lesson that data
quality and methodical consideration of selection bias are often much more
important than data quantity and specific models. This is especially
important in high-stakes decisions like the ones for the COVID-19
pandemic. In a world where data science is playing an ever-larger role
in policy making, ignoring selection bias could become the most
costly of bets.}


\section*{Acknowledgement}
We thank Cindy Chen, Yang Chen, Yunjin Choi, Hera He, Michael Levy, Marc
Lipsitch, James Robins, Andrew Rosenfeld, Dylan Small, Yachong Yang,
and Zilu Zhou for their helpful suggestions. We thank citizens living
in the first author's hometown, Wuhan, whose enduring adherence to
the travel quarantine not only saved many lives but also made our
analysis possible.


\bibliographystyle{imsart-nameyear}
\bibliography{ref}

\clearpage
\appendix
\section{Technical proofs}
\label{sec:technical-proofs}

\subsection{Derivation of \Cref{lem:selection-probability-general}}
\label{sec:proof-crefpr-prob}

Using \cref{eq:model-t,eq:model-s}, it is straightforward to show that
\begin{align*}
  &\P((B,E,T,S) \in \mathcal{D} \mid B = b, E = e) \\
  =& \P(b \le T \le e, T \le S < \infty \mid B = b, E = e) \\
  =& \int_{t \in (b,e)} f_T(t \mid b,e)
     \int_{s \in (t,\infty)} f_S(s \mid b,e,t)  \, ds \, dt \\
  =& \int_{t \in (b,e)} f_T(t \mid b,e)
     \Big\{\int_{s \in (t,\infty)} \nu \cdot h(s - t) \, ds\Big\} \, dt
  \\
  =&  \int_{t \in (b,e)} f_T(t \mid b,e)
     \cdot \nu \, dt \\
  =&  \nu [G(e) - G(b)].
\end{align*}

\subsection{Derivation of \Cref{lem:conditional-density-parametric}}
\label{sec:deriv-crefl-dens}

By \Cref{assump:model-t-parametric}, $G_{\kappa,r}(t) = \int_{-\infty}^t g_{\kappa,r}(s) \, ds =
(\kappa/r) \exp(r t)$. Thus for $b > 0$, we have
\begin{equation} \label{eq:integral-e}
  \begin{split}
    &\P((b,E,T,S) \in \mathcal{D} \mid B = b) \\
    =&\nu \int_{e \in (b, L)} f_E(e \mid b) \, \big[G_{\kappa,r}(e) - G_{\kappa,r}(b) \big]
    \, de \\
    =& \nu \int_b^L \lambda_V \, (\kappa / r) \, \{\exp(r e)
    - \exp(r b) \} \, de \\
    =& \frac{\lambda_V \kappa \nu}{r} \Big[\frac{1}{r} \big(\exp(rL) -
    \exp(r b)\big) - (L - b) \exp(rb)\Big] \\
    =& \frac{\lambda_V \kappa \nu}{r^2} \exp(rL) - \frac{\lambda_V
      \kappa \nu}{r}
    (r^{-1} + L - b) \exp(rb) \\
    =& \frac{\lambda_V \kappa \nu}{r^2} \exp(rL) \Big[1 - (1+r(L-b)) \exp(-r
    (L-b))\Big].
  \end{split}
\end{equation}
For $b = 0$, we can replace $\lambda_V$ in the above equation by $\lambda_W$.

The idea is that, if $rL$ is much larger than $1$ (in our preliminary
analysis $rL \approx 0.25 \times 54 = 13.5$), then
\[
  \text{Right hand side of}~\eqref{eq:integral-e} \approx \frac{\nu \lambda_W
    \kappa}{r^2} \exp(rL)~\text{when}~b = 0.
\]
Using this approximation, we obtain
\begin{align*}
    &\P((B,E,T,S) \in \mathcal{D}) \\
    =&\int_{0 \le b < L} \P((b,E,T,S) \in \mathcal{D} \mid B = b) \,
    f_B(b) \, db \\
    =& \P(B = 0) \cdot \P((b,E,T,S) \in \mathcal{D} \mid B = 0) + \int_{0 < b < L} \P((b,E,T,S) \in \mathcal{D} \mid B = b) \,
    f_B(b) \, db \\
    \approx& \frac{(1-\pi) \lambda_W \kappa \nu}{r^2} \exp(rL) + \int_{0 < b < L} \P((b,E,T,S) \in \mathcal{D} \mid B = b) \,
    f_B(b) \, db \\
    =& \frac{(1-\pi) \lambda_W \kappa \nu}{r^2} \exp(rL) + \pi \int_0^L
    \frac{1}{L} \frac{\lambda_V \kappa \nu}{r^2} \exp(rL) \Big[1 - (1+r(L-b)) \exp(-r
    (L-b))\Big] \, db \\
    =& \frac{(1-\pi) \lambda_W \kappa \nu}{r^2} \exp(rL) + \frac{\pi
      \lambda_V \kappa \nu}{r^2} \exp(rL) - \frac{\pi}{L} \frac{\lambda_V \kappa \nu}{r^2} \exp(rL) \underbrace{\int_0^L
      \Big[(1+r(L-b)) \exp(-r (L-b))\Big]}_{A_1} \, db \\
    \approx& \frac{\kappa \exp(rL) \nu}{r^2} \Big[(1-\pi) \lambda_W + \pi
    \lambda_V (1 - 2/(rL)) \Big].
\end{align*}
In the last step we used the approximation $e^{rL} \gg 1 + rL$:
\[
  A_1 = \int_0^L (1 + r x) \exp(-r x) \, dx
  =
  - \frac{\exp(-rx) (rx + 2)}{r} \Big|_{x=0}^{x=L}
  = \frac{2}{r} - \frac{\exp(-rL)(rL+2)}{r}
  \approx \frac{2}{r}.
\]

Therefore, the density is given by
\[
  \begin{split}
    f(b,e,t,s \mid \mathcal{D}) \approx& \frac{[(1 - \pi) \lambda_W 1_{\{b = 0\}} +
      (\pi/L) \lambda_V 1_{\{b > 0\}}] \cdot \kappa
      \exp(rt) \cdot \nu h(s-t)}{r^{-2}\kappa \exp(rL) \nu \Big[(1-\pi) \lambda_W + \pi
      \lambda_V (1 - 2/(rL)) \Big]} \\
    =& r^2 \cdot \frac{[(1 - \pi) \lambda_W 1_{\{b = 0\}} +
      (\pi/L) \lambda_V 1_{\{b > 0\}}] \cdot \exp(rt)}{\Big[(1-\pi) \lambda_W + \pi
      \lambda_V (1 - 2/(rL)) \Big] \cdot \exp(rL)} \cdot h(s-t) \\
    =& r^2 \cdot \frac{[ 1_{\{b = 0\}} +
      (\rho/L) 1_{\{b > 0\}}] \cdot \exp(rt)}{\Big[1 + \rho (1 - 2/(rL)) \Big] \cdot \exp(rL)} \cdot h(s-t), \\
  \end{split}
\]
where $\rho = (\lambda_V / \lambda_W) \pi / (1-\pi)$.

\subsection{Derivation of
  \Cref{prop:observed-data-likelihood-parametric}}
\label{sec:deriv-crefpr-data}

The following Lemma is useful to marginalize over $T$ when the incubation
period follows a $\text{Gamma}(\alpha,\beta)$ distribution:
\begin{lemma} \label{lem:gamma-integral}
  For any $r > 0$ and $b \le e \le s$,
  \[
    \int_b^{\min(s,e)} \exp(rt) \,h_{\alpha,\beta}(s-t) \, dt =
    \Big(\frac{\beta}{\beta+r}\Big)^{\alpha} \exp(r s) \,
    \big[H_{\alpha,\beta+r}(s-b) - H_{\alpha,\beta+r}((s-e)_+)\big].
  \]
\end{lemma}
\begin{proof}
  By a change of variables,
  \[
    \begin{split}
      &\int_b^{\min(s,e)} \exp(rt) \, h_{\alpha,\beta}(s-t) \, dt \\
      =&
      \int_b^{\min(s,e)} \exp(rt) \, \frac{\beta^{\alpha}}{\Gamma(\alpha)} (s-t)^{\alpha-1}
      \exp\{-\beta (s-t)\} \, dt \\
      =&  \Big(\frac{\beta}{\beta+r}\Big)^{\alpha} \exp(r s)
      \int_b^{\min(s,e)} \frac{(\beta+r)^{\alpha}}{\Gamma(\alpha)} (s-t)^{\alpha-1}
      \exp\{-(\beta+r)(s-t)\} \, dt \\
      =&  \Big(\frac{\beta}{\beta+r}\Big)^{\alpha} \exp(r s) \,
      \big[H_{\alpha,\beta+r}(s-b) - H_{\alpha,\beta+r}((s-e)_+)\big].
    \end{split}
  \]
\end{proof}

The time of contraction $T$ is not observed. Should it be observed,
the full data unconditional likelihood is given by
\begin{align*}
    &L_{\text{uncond}}\big(\rho, r, h(\cdot); \bm T\big)\\
    =&\prod_{i=1}^n f(B_i,E_i,T_i,S_i \mid
    (B_i,E_i,T_i,S_i) \in \mathcal{D}) \\
    \approx& r^{2n} \cdot \prod_{i=1}^n\frac{ 1_{\{B_i = 0\}} +
      (\rho/L) 1_{\{B_i > 0\}} }{1 + \rho (1 - 2/(rL))} \cdot
    \prod_{i=1}^n \underbrace{ 1_{\{B_i \le T_i \le
        \min(E_i,S_i)\}} \cdot
      \exp(r(T_i - L)) \cdot h(S_i - T_i)}_{A_{2,i}}.
\end{align*}

If we assume $h(\cdot)$ is the density of a Gamma distribution:
\[
  h(x) = h_{\alpha,\beta}(x) = \frac{\beta^{\alpha}}{\Gamma(\alpha)} x^{\alpha - 1}
  \exp(-\beta x) \quad (x > 0 ) ,
\]
then we can marginalize over $T_i$ using \Cref{lem:gamma-integral}:
\[
    \int A_{2,i} \, d T_i = \exp\big\{r (S_i - L)\big\}
    \Big(\frac{\beta}{\beta+r}\Big)^{\alpha} \cdot
    \big[H_{\alpha,\beta+r}(S_i - B_i) -
    H_{\alpha,\beta+r}((S_i-E_i)_{+}) \big],
\]

In conclusion, the unconditional observed data likelihood is given by
\[
  \begin{split}
    L_{\text{uncond}}(\rho,r,\alpha,\beta)
    \approx r^{2n} \Big(\frac{\beta}{\beta+r}\Big)^{n\alpha} \cdot &\prod_{i=1}^n\bigg\{\frac{1_{\{B_i = 0\}} +
      (\rho/L) 1_{\{B_i > 0\}} }{1 + \rho (1 - 2/(rL))} \\
    & \times \exp\big\{r (S_i - L)\big\}
    \big[H_{\alpha,\beta+r}(S_i - B_i) -
    H_{\alpha,\beta+r}((S_i-E_i)_{+}) \big]\bigg\}.
  \end{split}
\]
The conditional observed data likelihood can be derived in the same
way. Details are omitted.

\subsection{Derivation of \Cref{prop:marginal-s}}
\label{sec:deriv-crefpr-s}

By integrating the conditional density
\eqref{eq:conditional-density-parametric} over $(b,e,s)$, the marginal
distribution of $T$ conditional on $(B,E,T,S) \in \mathcal{D}$ is
given by
\revv{
\[
  \begin{split}
  f_T(t \mid \mathcal{D}\revv{,B=0}) &\propto \int \int \int f(b,e,t,s
  \mid \mathcal{D}) \cdot 1_{\{(b,e,t,s) \in \mathcal{D}\}} \cdot 1_{\{b=0\}} \, db \, de \, ds \\
  &\approx \int_0^t \int_t^L \int_t^{\infty} r^2 \cdot \frac{1_{\{b = 0\}}  \cdot \exp(rt)}{\Big[1 + \rho (1 -
    2/(rL)) \Big] \cdot \exp(rL)} \cdot h(s-t) \, ds \, de \, db \\
  &\propto \int_t^L \int_t^{\infty} \exp(rt) \, h(s-t) \, ds \, de\\
  &= \int_t^L \exp(rt) \, de\\
  &= (L-t) \exp(rt).
  \end{split}
\]
}


\Cref{assump:model-s} says that the distribution of the symptom onset
$S$ only depends \rev{ on} the time of transmission $T$ ($S \independent
B, \rev{ E} \mid T$). Therefore the marginal distribution of $S$ in
exported Wuhan resident cases is given by convolving the distribution of $T$
with the distribution of the incubation period $S-T$:
\[
  f_S(s \mid \mathcal{D}\revv{,B=0}) = \int_0^{\min(L,s)} f_T(t \mid \mathcal{D}\revv{,B=0}) h(\revv{s-t})
  \, dt
\]
Under the parametric assumption that $S-T$ follows a Gamma
distribution (\Cref{assump:model-s-parametric}), we have,
\rev{for $s \ge L/2$,}
\[
  \begin{split}
  f_S(s \mid \mathcal{D}\revv{,B=0}) &\appropto \int_0^{\min(L,s)} (L-t) \exp(rt) \cdot
  (s-t)^{\alpha-1} \exp \{-\beta(s-t)\} \, dt \\
  &= \exp(rs) \cdot \int_0^{\min(L,s)} [(L-s) + (s-t)] (s-t)^{\alpha-1}
  \exp\{-(\beta+r)(s-t)\} \, dt \\
  &= \exp(rs) \cdot \int_{(s-L)_+}^s [(L-s) x^{\alpha - 1} +
  x^{\alpha}] \exp\{-(\beta+r) x\} \, dx \\
  &= \exp(rs) \cdot \Big\{ (L-s)
  \frac{\Gamma(\alpha)}{(\beta+r)^{\alpha}} \big[H_{\alpha,\beta+r}(s)
  - H_{\alpha,\beta+r}((s-L)_+) \big] \\
  &\qquad \qquad \qquad +
  \frac{\Gamma(\alpha+1)}{(\beta+r)^{\alpha+1}} \big[H_{\alpha+1,\beta+r}(s)
  - H_{\alpha+1,\beta+r}((s-L)_+) \big] \Big\} \\
  &\propto \exp(rs) \cdot \Big\{ (L-s) \big[H_{\alpha,\beta+r}(s)
  - H_{\alpha,\beta+r}((s-L)_+) \big] \\
  &\qquad \qquad \qquad + \frac{\alpha}{\beta+r} \big[H_{\alpha+1,\beta+r}(s)
  - H_{\alpha+1,\beta+r}((s-L)_+) \big] \Big\}. \\
  &\approx \exp(rs) \cdot \Big\{ (L-s)
  [1-H_{\alpha,\beta+r}((s-L)_+)] + \frac{\alpha}{\beta + r}[1 -
  H_{\alpha+1,\beta+r}((s-L)_+)]\Big\}.
  \end{split}
\]
\rev{
The last step uses the approximation that
\[
  1 \approx H_{\alpha+1,\beta+r}(L/2) \le H_{\alpha+1,\beta+r}(s) \le
  H_{\alpha,\beta+r}(s).
\]
We next show that this is reasonable under the
technical assumption $L > 4(\alpha+5)/(\beta+r)$ in
\Cref{prop:marginal-s}. Suppose $X \sim
\text{Gamma}(\alpha+1,\beta+r)$. The Chernoff tail bound says that
\[
  1 - H_{\alpha+1,\beta+r}(L/2) = \P(X > L/2) \le \frac{\E[\exp(c
    X)]}{\exp(c L/2)} = \frac{(1 -
    c/(\beta+r))^{-(\alpha+1)}}{\exp(cL/2)}~\text{for}~a < \beta+r.
\]
By choosing $c = (1-\exp(-1))(\beta+r) \approx 0.63 (\beta+r)$, we have
\[
  1 - H_{\alpha+1,\beta+r}(L/2) \le \frac{\exp(\alpha+1)}{\exp\{0.31
    (\beta+r)L\}} < \frac{\exp(\alpha+1)}{\exp\{0.31 \times 4
    (\alpha+5)\}} < \exp(1 - 0.31 \times 20) < 0.01.
\]
}

\subsection{Derivation of \Cref{prop:likelihood-truncation}}
\label{sec:deriv-crefpr-trunc}

Let $e_- = \min(e,M)$. Then under \Cref{assump:model-t,assump:model-s}, for
$(b,e,t,s) \in \mathcal{D}$ and $s \le M$,
\begin{equation} \label{eq:appendix-1}
  \begin{split}
    f_{T,S}(t,s \mid b,e,\mathcal{D}, S \le M) &= \frac{f_{T,S}(t,s \mid
      b,e,\mathcal{D}) }{\int \int f_{T,S}(t,s \mid
      b,e,\mathcal{D})\,ds\,dt} \\
    &= \frac{g(t) h(s-t)}{\int_b^{e_-} g(t) \int_t^M  h(s-t) \, ds \, dt} \\
    &= \frac{g(t) h(s-t)}{\int_b^{e_-} g(t) H(M-t) \, dt},
  \end{split}
\end{equation}
where $H(s) = \int_0^s h(x) \, dx$ is the distribution function of the
incubation period. Assuming $g(t) = \kappa \exp(rt)$ and using
integration by parts, for $r \neq 0$,
\[
  \begin{split}
    \int_b^{e_-} g(t) H(M-t) \, dt =& \kappa \int_{b}^{e_-} \exp(rt) H(M-t)
    \, dt \\
    =&\frac{\kappa}{r} \int_b^{e_-} H(M-t) \, d \exp(rt) \\
    =&\frac{\kappa}{r} \Big[\exp(rt) H(M-t) \Big|_{t=b}^{t=e_-} +
    \int_b^{e_-} \exp(rt) h(M-t) \, dt\Big].\\
  \end{split}
\]
By using $h(\cdot) = h_{\alpha,\beta}(\cdot)$ and using
\Cref{lem:gamma-integral}, we have
\[
  \int_b^{e_-} g(t) H_{\alpha,\beta}(M-t) \, dt = \frac{\kappa}{r}
  \Big[ \exp(rt) H_{\alpha,\beta}(M-t) - \Big(\frac{\beta}{\beta+r}\Big)^{\alpha} \exp(rM)
  H_{\alpha,\beta+r}(M-t) \Big]\bigg|_{t = b}^{t = e_-}.
\]
Now we integrate $t$ in \eqref{eq:appendix-1} from $b$ to $e_{-}$ and get
\[
  \begin{split}
    &f_S(s \mid b,e,\mathcal{D},S \le M) \\
    =&
   \frac{r \Big(\frac{\beta}{\beta+r}\Big)^{\alpha} \exp(r s) \,
    \big[H_{\alpha,\beta+r}(s-b) -
    H_{\alpha,\beta+r}((s-e)_{+})\big]}{\Big[ \exp(rt) H_{\alpha,\beta}(M-t) - \Big(\frac{\beta}{\beta+r}\Big)^{\alpha} \exp(rM)
  H_{\alpha,\beta+r}(M-t) \Big]\bigg|_{t = b}^{t = e_-}}.
  \end{split}
\]

For $r = 0$, using integration by parts,
\[
  \begin{split}
    \int_b^{e_-} g(t) H_{\alpha,\beta}(M-t) \, dt
    =& \kappa \int_{(M-e)_+}^{M-b} H_{\alpha,\beta}(x) \, dx \\
    =& \kappa \Big[ x H_{\alpha,\beta}(x)
    \Big|_{x=(M-e)_+}^{x=M-b} - \int_{(M-e)_+}^{M-b} x h_{\alpha,\beta}(x) \,
    dx \Big] \\
    =& \kappa \Big[ x H_{\alpha,\beta}(x)
    \Big|_{x=(M-e)_+}^{x=M-b} - \int_{(M-e)_+}^{M-b} x \cdot
    \frac{\beta^{\alpha}}{\Gamma(\alpha)} x^{\alpha-1} \exp(-\beta x)
    \, dx \Big] \\
    =& \kappa \Big[ x H_{\alpha,\beta}(x)
    \Big|_{x=(M-e)_+}^{x=M-b} - \frac{\alpha}{\beta} \int_{(M-e)_+}^{M-b}
    \frac{\beta^{\alpha+1}}{\Gamma(\alpha+1)} x^{\alpha} \exp(-\beta x)
    \, dx \Big] \\
    =& \kappa \Big[ x H_{\alpha,\beta}(x)
    - \frac{\alpha}{\beta} H_{\alpha+1,\beta}(x) \Big]
    \bigg|_{x=(M-e)_+}^{x=M-b}.
  \end{split}
\]
We can similarly integrate $t$ out and obtain the full data
likelihood. Details are omitted.

\end{document}